\documentclass[a4paper,UKenglish,cleveref, autoref, thm-restate]{lipics-v2021}
\hideLIPIcs  %uncomment to remove references to LIPIcs series (logo, DOI, ...), e.g. when preparing a pre-final version to be uploaded to arXiv or another public repository

\nolinenumbers

\usepackage{latexsym,amsmath,textcomp,pifont,marvosym,wasysym,amssymb}
\usepackage{graphicx}
\usepackage{xcolor}
\usepackage{mathtools}
\usepackage{booktabs}
\usepackage{bbold}

\usepackage{numprint}
\usepackage{hyperref}

\usepackage{tikz}
\usetikzlibrary{calc}

\usepackage{pgfplots}
\usepgfplotslibrary{external}
\tikzsetexternalprefix{plots/}

\newif\ifpdfplots
\pdfplotstrue
\usepackage{algorithmic}
\usepackage[algo2e, ruled, vlined, linesnumbered]{algorithm2e}
\DontPrintSemicolon
\SetKwProg{Fn}{Function}{}{end}
\SetKwFor{ParallelWhile}{while}{do in parallel}{}
\SetKwFor{ParallelDo}{parallel do}{}{}
\SetKwFor{ParallelFor}{for}{do parallel}{}
\SetKwInput{KwInput}{Input}
\SetKwInput{KwOutput}{Output}

\title{High-Quality Multi-Constraint Hypergraph Partitioning via Greedy Rebalancing}

\author{Nikolai Maas}{Karlsruhe Institute of Technology, Karlsruhe, Germany}{nikolai.maas@kit.edu}{https://orcid.org/0009-0002-6959-417X}{}

\authorrunning{N. Maas}

\Copyright{Nikolai Maas} %TODO mandatory, please use full first names. LIPIcs license is "CC-BY";  http://creativecommons.org/licenses/by/3.0/

\ccsdesc[500]{Theory of computation~Design and analysis of algorithms}

\keywords{Hypergraph Partitioning, Multi-Constraint Partitioning, Graph Algorithms, Multilevel Algorithms, Local Search, Vector Scheduling, Multidimensional Load Balancing}

\category{} %optional, e.g. invited paper

\relatedversion{} %optional, e.g. full version hosted on arXiv, HAL, or other respository/website
\supplement{}
\supplementdetails[subcategory={}, cite={}, swhid={}]{Source Code}{https://github.com/kahypar/mt-kahypar/tree/esa2026}
\supplementdetails[subcategory={}, cite={}, swhid={}]{Benchmark Set \& Experimental Results}{https://zenodo.org/records/19630135}

\newcommand{\argmax}{\operatorname{arg \, max}}

\newcommand{\gainfn}[1]{\ensuremath{\Delta #1}}

\newcommand{\norm}[1]{\ensuremath{\|#1\|}}

\newcommand{\dimindex}{\ensuremath{j}}

\definecolor{darkgreen}{rgb}{0.0, 0.6, 0.0}

\newcommand{\kommentar}[1]{}
\newcommand{\nikolai}[1]{\kommentar{\color{darkgreen}[NM: #1]}}

\newcommand{\review}[1]{\kommentar{\color{violet}[R: #1]}}

\newcommand{\splitatcommas}[1]{%
  \begingroup
  \begingroup\lccode`~=`, \lowercase{\endgroup
    \edef~{\mathchar\the\mathcode`, \penalty0 \noexpand\hspace{0pt plus 1em}}%
  }\mathcode`,="8000 #1%
  \endgroup
}

\newcommand{\Partitioner}[1]{#1}
\newcommand{\Instance}[1]{\textsf{#1}}

\newcommand{\etal}{et~al.}

\newcommand{\mss}[1]{%
  \ifcase#1\relax
    \or Set~B$_{1\textnormal{-}20}$%  Case 1
    \or Set~B$_{21\textnormal{-}45}$% Case 2
    \or Set~B$_{46\textnormal{-}69}$% Case 3
    \else [Invalid Input]% 
  \fi
}

\begin{document}

\maketitle

\begin{abstract}
    Multi-constraint hypergraph partitioning is a generalization of balanced partitioning, where the vertex set of a hypergraph is partitioned such that the inter-block connectivity of hyperedges is minimized while balancing the vertices with regard to $d$ distinct constraints.
    A prominent class of applications is data distribution tasks, where this allows to achieve good load balance for $d$ different kinds of resources and simultaneously minimize the communication volume.
    %Yet, research on multi-constraint partitioning is sparse compared to the single-constraint case.

    Although the best approaches for single-constraint partitioning are usually complex (multilevel) algorithms with many components,
    we show that replacing only one component already leads to high-quality multi-constraint partitions:
    the rebalancing step, which restores balance for a partition that has (hopefully) small connectivity but violates the constraints.
    We design a multi-constraint rebalancing algorithm based on greedy local search, proving that balance is always restored for $d=2$ and bounded maximum weight.
    The key is to ensure monotonically decreasing global imbalance by choosing an imbalance metric where there is always a balance-improving move available.
    Integrating our algorithm into the state-of-the-art partitioner \Partitioner{Mt-KaHyPar},
    we demonstrate an 11.5\,\% geometric mean connectivity reduction compared to the next best competitor (\Partitioner{Metis}) and better reliability regarding partition balance,
    even though the majority of inputs is outside of the theoretical guarantee.
\end{abstract}

\newpage

\section{Introduction}

Traditional balanced hypergraph partitioning asks for a partition of the vertices of a hypergraph $H = (V, E)$ into $k$ disjoint blocks $V_1, \dots, V_k \subseteq V$ of roughly equal size $|V_i| \le (1 + \varepsilon) \frac{|V|}{k}$ (the balance constraint) while minimizing the inter-block connectivity of the hyperedges.
This is an essential task for many applications, including distributed databases~\cite{UnifyDR, schism}, VLSI circuit design~\cite{vlsiDuttD02, vlsiGuoLCP14}, and scientific simulations~\cite{fluidSimulations}.
However, a single balance constraint is often insufficient to model application requirements.
For example, multi-phase simulations balance the compute load in each phase while minimizing dependencies between processors.
If phases are treated as separate partitioning problems, expensive data migrations are required in between---\nolinebreak
which can be avoided by computing a global partition with multiple distinct balance constraints~\cite{temporalAdaptiveMeshPartitioning, gpMultiPhaseSimulations}.
Another application is distributed processing tasks that need to balance both compute and memory requirements.
As a concrete example, high-quality partitions significantly improve the efficiency of distributed GNN training~\cite{neuGraphGNN, PipeGCN}.

Multi-constraint partitioning models these use cases with $d$-dimensional vertex weights, thereby expressing $d$ different kinds of balance constraints at once.
Despite its practical relevance, there is not much existing work on multi-constraint partitioning.
Current approaches mostly use straightforward adjustments of single-constraint partitioning algorithms and restrict local search procedures to disregard moves that violate any of the $d$ constraints~\cite{kPaToH, MulticonstraintScotch, MulticonstraintMetis, MulticonstraintParMetis}.
We believe that these \emph{constrained} search algorithms have a severe limitation;
since $d$ balance constraints could shrink the subspace of feasible solutions significantly, it becomes hard for the search to escape from local minima.
Instead, we propose to use \emph{unconstrained} search algorithms that allow temporary balance violations,
guiding the search through regions of infeasible solutions instead of around them.
% This is partially motivated by the fact that dual approaches which start from infeasible (but high-quality) solutions are a well-known and powerful tool in the wider field of constrained optimization \nikolai{elaborate}.
%This provides more options to escape from local minima by crossing the space of infeasible solutions instead of searching around it.
We already achieved improvements for traditional partitioning via unconstrained search techniques~\cite{mt-kahypar-detv2, mt-kahypar-ufm}, with similar findings by other authors~\cite{Jet-Journal, kaminpar-distributed-unconstrained}.
In this work, we will show the effectiveness of this technique in the multi-constraint setting.

However, allowing balance violations means that the algorithm must restore balance afterwards, thereby alternating between feasible and infeasible solutions.
This step is known as \emph{rebalancing} and it is a much more difficult task in the multi-constraint setting than with a single constraint~\cite{MulticonstraintMetis}.
At its core, it is equivalent to a multi-dimensional job scheduling problem---but we must simultaneously consider the optimization objective.
We therefore develop a rebalancing algorithm that efficiently and reliably restores partition balance.
The algorithm quantifies the current imbalance with a metric based on the $L_1$-norm
and greedily performs single-vertex moves that minimize the impact on quality while monotonically reducing imbalance.
In addition, we use a secondary heuristic to handle cases where the greedy approach can not make progress.
It moves a small set of vertices in a way that reduces block-internal imbalances between dimensions, thereby unlocking new move options.

\subparagraph{Contributions.}
We propose a new rebalancing algorithm for multi-constraint partitioning and integrate it into the state-of-the-art parallel hypergraph partitioner Mt-KaHyPar~\cite{mt-kahypar-journal},
while also implementing multi-constraint support for all components of Mt-KaHyPar.
We prove that our rebalancing always restores balance for $d=2$ and bounded maximum weights.
Experiments on 77 hypergraph instances and 67 graph instances show that our algorithm still finds a feasible solution with over 99\,\% reliability for $d > 2$ and larger weights,
whereas state-of-the-art competitors are below 90\,\%.
Our algorithm computes the best solution for 79\,\% of instances,
achieves an 11.5\,\% geometric mean quality improvement compared to the next best competitor,
and it is the fastest multilevel algorithm when using 16 threads.

\section{Preliminaries}\label{sec:preliminaries}

\subparagraph*{Hypergraphs and Multi-Constraint Partitioning.}
For a given dimensionality $d \in \mathbb{N}$,
let $H = (V, E, c, \omega)$ be a \emph{weighted hypergraph} with a vertex set $V$, a set of hyperedges $E \subseteq \mathcal{P}(V)$, multi-dimensional vertex weights $c\colon V \mapsto \mathbb [0, 1]^d$, and edge weights $\omega\colon E \mapsto \mathbb{R}_{> 0}$.
We extend $c$ and $\omega$ to sets, i.e., $c(V' \subseteq V) \coloneqq \sum_{v \in V'} c(v)$ and $\omega(E' \subseteq E) \coloneqq \sum_{e \in E'} \omega(e)$.
The vertices of a hyperedge $e$ are called the \emph{pins} of $e$.
%A vertex $v$ is \emph{incident} to $e$ if $v \in e$ and $I(v)$ denotes the set of incident hyperedges of $v$.
%The \emph{degree} of $v$ is $d(v) \coloneqq |I(v)|$.
A \emph{$k$-way partition} $\Pi = \{V_1, \dots, V_k\}$ of a hypergraph $H$ is a set of $k$ disjoint \emph{blocks} that partition $V$, i.e., $V = V_1 \cup \cdots \cup V_k$.
When considering a $k$-way partition $\Pi$, we always assume that the vertex weights are normalized such that $\frac{1}{k} c(V)_\dimindex = 1$ for each dimension $\dimindex \in [d]$, as this simplifies notation.
We say that $\Pi$ is \emph{$\varepsilon$-balanced} if $c(V_i)_\dimindex \le 1 + \varepsilon$ for each $i \in [k]$ and $\dimindex \in [d]$.
The \emph{connectivity} of a hyperedge $e$ is $\lambda(e) \coloneqq |\{V_i \in \Pi \mid V_i \cap e \ne \emptyset \}|$.
%The \emph{connectivity set} of a hyperedge $e$ is $\Lambda(e) \coloneqq \{V_i \mid V_i \cap e \ne \emptyset \}$ and the \emph{connectivity} of $e$ is $\lambda(e) \coloneqq |\Lambda(e)|$.
%\subparagraph*{Multi-Constraint Partitioning.}
%Given a hypergraph $H$ with $d$ weight dimensions and parameters
%Given $k$ and $\varepsilon$,
\emph{Multi-constrained balanced hypergraph partitioning} asks for an $\varepsilon$-balanced $k$-way partition of $H$ that minimizes the \emph{connectivity} $(\lambda - 1)(\Pi) \coloneqq \sum_{e \in E} (\lambda(e) - 1) \omega(e)$.
We also refer to the connectivity as the solution quality of the partition.
The problem of finding any balanced partition, ignoring quality, is equivalent to the vector scheduling problem with $k$ identical machines.
Vector scheduling admits an EPTAS for constant $d$~\cite{VectorSchedulingBounds} (with doubly-exponential dependence on $d$), but constant factor approximations for arbitrary $d$ are hard~\cite{VectorSchedulingGeneral}.

\subparagraph*{Multilevel Algorithms.}
A common approach for (hyper-)graph partitioning in practice is the multilevel scheme.
Multilevel algorithms construct a sequence of increasingly smaller hypergraphs by contracting clusters or matchings (\emph{coarsening phase}),
followed by computing an \emph{initial partition} on the smallest hypergraph.
Afterwards, the contractions are undone in reverse order while projecting the current partition to the next larger hypergraph (\emph{uncoarsening phase}).
At each step, the partition is further improved via local search algorithms (\emph{refinement}).
While approximating balanced partitioning within a constant factor is NP-hard in general~\cite{AndreevRaeckeApprox},
multilevel algorithms often find high-quality solutions in practice~\cite{HYPERGRAPH-SURVEY, mt-kahypar-journal}.

\section{Related Work}\label{sec:related}

\subparagraph*{Multilevel Partitioning.}
We refer to recent surveys~\cite{GRAPH-RECENT-SURVEY, HYPERGRAPH-SURVEY} for a broad overview on hypergraph partitioning and multilevel partitioning in particular.
Within the multilevel framework, strong refinement algorithms are crucial for improving the partition during uncoarsening and thereby obtaining a high-quality solution.
The FM algorithm by Fiduccia and Mattheyses~\cite{FM} is one of the most widely used refinement algorithms as it offers a good trade-off between quality and running time.
It greedily moves vertices prioritized by their current gain, allowing moves with negative gains in order to escape local minima (afterwards rolling back to the best observed solution).
However, recent work~\cite{Jet-Journal,mt-kahypar-detv2,mt-kahypar-ufm,kaminpar-distributed-unconstrained} demonstrated that it is beneficial to also allow temporary balance violations.
Refinement algorithms based on this principle include Jet~\cite{Jet-Journal}, as well as unconstrained variants of FM and label propagation refinement~\cite{mt-kahypar-ufm}.
The effectiveness of these approaches depends on the rebalancing algorithm used to restore partition balance, which should minimize the incurred quality penalty~\cite{mt-kahypar-ufm}.

Our implementation extends \Partitioner{Mt-KaHyPar}~\cite{mt-kahypar-journal, mt-kahypar-d} (default configuration) for multi-constraint partitioning.
\Partitioner{Mt-KaHyPar} is a scalable shared-memory hypergraph partitioner with quality comparable to the best sequential algorithms~\cite{mt-kahypar-journal},
full support for unconstrained refinement~\cite{mt-kahypar-ufm},
multiple configurations with different time to quality trade-offs~\cite{mt-kahypar-journal, mt-kahypar-flows, mt-kahypar-q},
and specialized configurations that guarantee deterministic results~\cite{mt-kahypar-sdet, mt-kahypar-detv2}.

\subparagraph*{Multi-Constraint Partitioning.}
In comparison to single-constraint partitioning, literature on the multi-constraint setting is sparse.
Metis supports multi-constraint partitioning on graph inputs~\cite{MulticonstraintMetis},
with key components including an adjusted tie breaking during coarsening (\emph{balanced-edge} heuristic)
and a multi-queue approach with $d$ priority queues for vertex moves.
During bipartitioning and refinement, this allows to select vertices whose weight is a good fit for the current partition balance.
Metis uses rebalancing if the initial solution is imbalanced, but forbids balance-violating moves once balance is achieved.
These techniques were then transferred to ParMetis~\cite{MulticonstraintParMetis, MulticonstraintParMetisDynamic}, with new contributions focusing on a scalable two-pass refinement algorithm.
The idea is to compute a set of move candidates and then select a subset that avoids (too large) balance violations.
The authors observed that small imbalances can be beneficial for solution quality~\cite{MulticonstraintParMetis}.
For the Scotch partitioner, Barat \etal\ implement multi-constraint support via ``small variations'' of the baseline algorithm~\cite{MulticonstraintScotch}.
However, the implementation is not public.

Some more recent graph partitioning algorithms operate outside of the multilevel framework.
This includes GD~\cite{GradientDescentPartitioning}, a multi-constraint partitioner based on projected gradient descent by Avdiukhin \etal, and PuLP~\cite{Pulp}, a single-level graph partitioner designed for scalability.
PuLP does not implement general multi-constraint partitioning, but it supports the special case of balancing both node weights and edge weights simultaneously.

For hypergraphs,
Karypis~\cite{MulticonstraintHMetisTR} applies ideas from multi-constraint Metis to the hMetis hypergraph partitioner.
Selvakkumaran \etal~\cite{MulticonstraintHMetisVLSI} evaluate bipartitioning for resource-aware FPGA placement, demonstrating good performance for approaches that apply postprocessing (i.e., rebalancing) to a hMetis partition.
Yet, both results are not integrated into the public hMetis package.\footnote{
    Testing this, it seems that hMetis can read a hypergraph file with multi-dimensional weights, but the partitioning algorithm only considers the balance of the first dimension.
}
Zoltan~\cite{ZOLTAN} is a hypergraph partitioner that includes multiple different partitioning approaches.
One of them is recursive coordinate bisection (RCB), a geometric mesh partitioning algorithm which supports multi-dimensional weights.
However, Zoltan has no support for multi-constraint partitioning without geometric information.
To the best of our knowledge, there are only two public implementations of general multi-constraint hypergraph partitioning.
PaToH~\cite{kPaToH, PATOH} implements multi-constraint partitioning by generalizing the standard algorithmic components,
but does not use any specialized algorithms.
%However, the authors observe that direct $k$-way partitioning is preferable to recursive bipartitioning~\cite{kPaToH}.
TritonPart~\cite{TritonPart} is designed for VLSI applications, including multi-dimensional weights.
However, the publication only evaluates single-constraint inputs, so its actual capabilities are unclear.

\section{Rebalancing for Multi-Constraint Partitioning}\label{sec:algorithm}

The primary design goals for our rebalancing algorithm are that
(i) it should restore balance with high reliability
and
(ii) it should minimize the impact on partition connectivity (so it finds connectivity improvements when combined with unconstrained refinement).
For the second goal, it is generally preferable to move only a small number of vertices.
% This makes it natural to use a local search that iteratively moves single vertices until balance is restored.
The remainder of this section therefore develops new techniques to restore balance by iteratively moving single vertices,
which becomes more difficult with increasing dimension $d$.

\begin{figure}
    \includegraphics{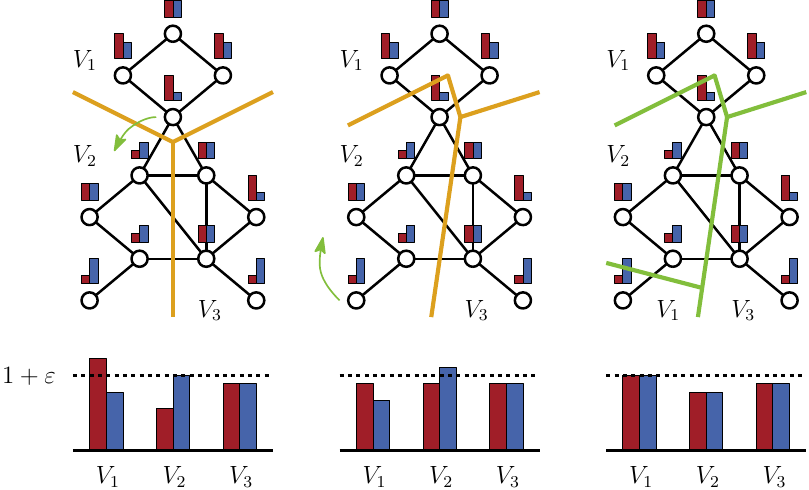}
    \caption{
        Illustration of rebalancing via the $L_1^u$ imbalance.
        Any move away from $V_1$ also overloads the target block (left).
        However, there is a sequence of two moves (applied from left to right) that monotonically decreases the $L_1^u$ imbalance (with $u = 1 + \varepsilon$) and results in a balanced partition.
    }
    \label{fig:rebalance_illustration}
\end{figure}

\subsection{A More General Imbalance Metric}
For the following discussion, remember that we assume all vertex weights are normalized such that the average block weight is 1 in each dimension (Section~\ref{sec:preliminaries}).
We further use $\delta \coloneqq \max_{v \in V, \dimindex \in [d]} c(v)_\dimindex$ to denote the maximum vertex weight.
For balanced partitioning, a common assumption is that weights are small ($\delta \le \varepsilon$)---\nolinebreak
otherwise, finding any feasible solution is already NP-hard~\cite{deep-mgp, KAHYPAR-BP}.
In the single-constraint case, this makes rebalancing a simple task; %, as it reduces to the classical parallel machine scheduling problem.
%Assuming that $\delta$ is an upper bound for the vertex weight,
a $\delta$-balanced assignment can be found by iteratively moving vertices from overloaded blocks to a block with load at most $1$.
%Therefore, engineering efforts focus on running time efficiency and quality---\nolinebreak
%the rebalancing algorithm should minimize the connectivity penalty incurred by the rebalancing.
However, this approach is already insufficient for $d = 2$.
There might be no block which simultaneously has small weight in both dimensions, and thus every possible move away from an overloaded block would overload another block (e.g., Figure~\ref{fig:rebalance_illustration}).
Instead, we propose to allow moves to overloaded blocks as long as they reduce the overall imbalance.
We formalize this with an imbalance metric based on the $L_1$ norm.

\begin{definition}\label{definition:l1_imbalance}
    For a given partition $\Pi = \{V_1, \dots, V_k\}$ of a hypergraph $H$ and an upper weight limit $u \in \mathbb{R}$,
    we define the $L_1^u$ imbalance of block $V_i$ as $L_1^u(V_i) \coloneqq \sum_{\dimindex \in [d]} \max\{c(V_i)_\dimindex - u, 0\}$ and the total $L_1^u$ imbalance as $L_1^u(\Pi) \coloneqq \sum_{i} L_1^u(V_i)$.
\end{definition}

We can observe that all blocks have maximum weight at most $u$ in all dimensions if and only if $L_1^u(\Pi) = 0$.
As illustrated in Figure~\ref{fig:rebalance_illustration},
the high-level idea of our rebalancing algorithm is to iteratively apply moves that reduce the $L_1^u$ imbalance, for an appropriately chosen value of $u$.
In the $2$-dimensional case, we can show that such a move always exists under certain conditions.

\begin{theorem}\label{theorem:2d}
    Let $d = 2$ and let $\Pi$ be a partition of a hypergraph $H$.
    If $u \ge 1 + \delta$ and there is a block $V_1 \in \Pi$ with $\norm{c(V_1)}_\infty > u + \delta$,
    then it is possible to move a vertex from $V_1$ to another block such that the $L_1^u$ imbalance is reduced.
\end{theorem}

\begin{proof}
    Without loss of generality, we assume $c(V_1)_1 > u + \delta$.
    Let $V_2$ be any block with $c(V_2)_1 \le 1$.
    First, we consider the case where $V_1$ contains a vertex $v$ with $c(v)_1 > c(v)_2$.
    Moving $v$ to $V_2$ does not overload the first dimension (due to $u \ge 1 + \delta$) and therefore increases the $L_1^u$ imbalance by at most $c(v)_2$.
    Combined with the reduced imbalance of $V_1$, the total change in imbalance is at most $c(v)_2 - c(v)_1 < 0$.
    In the case where $c(v)_2 \ge c(v)_1$ holds for all vertices $v \in V_1$, we get $c(V_1)_2 \ge c(V_1)_1 > u + \delta$ by summation.
    Let $v$ be a vertex assigned to $V_1$ with $c(v)_1 > 0$.
    Since removing $v$ from $V_1$ reduces the imbalance in both dimensions,
    the change in imbalance when moving $v$ to $V_2$ is at most $c(v)_2 - (c(v)_1 + c(v)_2) = -c(v)_1 < 0$.
\end{proof}

Using Theorem~\ref{theorem:2d}, we can construct a rebalancing algorithm that guarantees $2\delta$-balance.
Note that this is only a factor two worse than the one-dimensional bound of $\delta$.

\begin{corollary}\label{corollary:2d}
    If $d = 2$, then the following algorithm always results in a $2\delta$-balanced partition:
    While the partition is not $2\delta$-balanced, apply any move that reduces the $L_1^{1 + \delta}$ imbalance.
\end{corollary}

\begin{proof}
    The partition is not $2\delta$-balanced if there is a block $V_1$ with $\norm{c(V_1)}_\infty > 1 + 2\delta$,
    which implies with Theorem~\ref{theorem:2d} that a move decreasing the $L_1^{1 + \delta}$ imbalance always exists.
    The algorithm terminates since the sequence of imbalance values is strictly decreasing and the number of possible configurations is finite.
\end{proof}

Unfortunately, no similar guarantee holds for $d > 2$.
%Theorem~\ref{theorem:2d} essentially states that, for $d=2$ and $u \ge 1 + \delta$, it is always possible to move a vertex away from an overloaded block such that the $L_1^u$ imbalance decreases.
This is because the proof of Theorem~\ref{theorem:2d} relies on the following property:
in an overloaded block, we can find a vertex $v$ where the overloaded dimension dominates, i.e., $c(v)_1 > c(v)_2$.
Due to this, moving the vertex results in a negative imbalance delta of $c(v)_2 - c(v)_1$.
However, attempting the same for $d=3$ requires
$c(v)_1 > c(v)_2 + c(v)_3$.
It is clear we can not find a vertex with this property in the general case.
See also Appendix~\ref{sec:app_counterexample} for a counterexample that is independent of the proof.
However, we will show in Section~\ref{sec:experiments} that rebalancing with the $L_1^u$ imbalance still works well in practice.

\subsection{Greedy $L_1^u$ Rebalancing}\label{sec:l1_rebalancing}
For our full algorithm, we build on the greedy rebalancing scheme that has proven successful for single-constraint partitioning~\cite{mt-kahypar-ufm}.
Initially, greedy rebalancing selects a set of move candidates $M \subseteq V$ by collecting all vertices in overloaded blocks. % and evaluating possible target blocks.
For each vertex $v \in M$, a set of eligible target blocks $T(v) \subseteq \Pi$ is determined (in the single-constraint case, blocks that do not become overloaded when adding $v$) and the best target block is selected from $T(v)$ based on a rating function $r(v, \cdot)$.
These steps are straightforward to do in parallel.
Then, the algorithm iteratively moves the vertices in $M$, greedily selecting the vertex with the best rating in each step.
In practice, this is done with a concurrent relaxed priority queue~\cite{multiqueues-esa}, extracting the current best moves in parallel while updating the priorities of neighbors after a move is applied.
The rebalancing terminates when either balance is restored or all remaining moves in $M$ do not improve balance.

For multi-constraint rebalancing, we consider all moves that reduce the total $L_1^u$ imbalance---even if the target block is overloaded.
To define our rating function, we first need to introduce the \emph{gain} of a move.
Assuming we move a vertex $v$ to the block $V_i$,
let $\Pi'$ be the partition that results from applying this move.
We define the \emph{connectivity gain} as $\gainfn{(\lambda-1)}(v, V_i) \coloneqq (\lambda-1)(\Pi) - (\lambda-1)(\Pi')$
and the \emph{balance gain} analogously as $\gainfn{L_1^u}(v, V_i) \coloneqq L_1^u(\Pi) - L_1^u(\Pi')$.
Our rating then computes the ratio of connectivity gain to balance gain (high is better):

\[
r(v, V_i) \coloneqq \begin{cases}
    \frac{\gainfn{(\lambda-1)}(v, V_i)}{\gainfn{L_1^u}(v, V_i)}, & \text{if } \gainfn{(\lambda-1)}(v, V_i) < 0\\
    \gainfn{(\lambda-1)}(v, V_i) \cdot \gainfn{L_1^u}(v, V_i), & \text{if } \gainfn{(\lambda-1)}(v, V_i) \ge 0
\end{cases}
\]

This is effectively equivalent to the single-constraint case~\cite{mt-kahypar-ufm},
optimizing the rebalancing progress relative to the quality penalty (note that we expect a negative connectivity gain).
However, allowing moves that overload the target block also has drawbacks.
As vertices from newly overloaded blocks are not contained in $M$, the algorithm can not restore balance for these blocks.
Instead, we run up to ten repeated rounds of greedy rebalancing, stopping early if the balance is restored or a round is unable to improve the imbalance.\footnote{
    Including all vertices in $M$ is not a solution either, since changes to the partition balance require a (very expensive) global update of the target blocks---which is implicitly handled when using multiple rounds.
}

An important parameter is the choice of $u$.
A simple option is $u = 1 + \varepsilon$, but Theorem~\ref{theorem:2d} tells us to prefer $u = 1 + \varepsilon - \delta$, where $\delta$ is the maximum vertex weight.
However, real-world instances might contain some vertices with large weight relative to the average block weight.
Choosing $u$ too small because of this can cause problems, in particular since Theorem~\ref{theorem:2d} also requires that $u \ge 1 + \delta$ (and therefore $2\delta \le \varepsilon$).
Fortunately, it is often practically feasible to ignore a number of overly heavy vertices and effectively work with a smaller $\delta'$.
We use a parameterized threshold $t$ that limits the considered weight (defined as a fraction of the average block weight),
define $\delta'_\dimindex \coloneqq \min\{t, \max_v c(v)_\dimindex\}$ for $\dimindex \in [d]$
and choose per-dimension bounds
$u = (1 + \varepsilon - \delta'_1, \dots, 1 + \varepsilon - \delta'_d)^\top$;
extending Definition~\ref{definition:l1_imbalance} in the natural way to $u \in \mathbb{R}^d$.
Concretely, we found that $t=0.0025$ works well, see Section~\ref{sec:exp_ablation}.
During multilevel partitioning, rebalancing is applied at different levels of the coarsening hierarchy (Section~\ref{sec:preliminaries}).
We compute $\delta'$ separately for each level, since the maximum vertex weight might be different.

Moreover, race conditions can lead to the parallel execution of moves that worsen the imbalance in combination
(this is a non-issue for $d=1$ as overloading the target block can be avoided with atomic operations).
To prevent this, we add a rollback that restores the intermediate state with lowest imbalance after each round, based on a sequential move order.
This is equivalent to the rollback of FM refinement, but using the $L_1^u$ imbalance as objective.

\subsection{A Fallback to Fix Internal Imbalance}\label{sec:fallback}

If $d > 2$ or if heavy vertices are present, the $L_1^u$ imbalance still exhibits imbalanced local minima:
states where some blocks are overloaded, yet no single move exists that improves the $L_1^u$ imbalance.
This is usually due to \emph{internal imbalance} of blocks, i.e., one dimension has larger weight than the remaining dimensions.
With high internal imbalance, a block might be overloaded in one dimension although its weight is below average if summed over all dimensions.
We address this with a fallback heuristic designed to reduce internal imbalance,
which is triggered if a rebalancing round finds no balance-improving move.
The heuristic selects a small subset of vertices $S_i$ from each overloaded block $V_i$, which is moved to other blocks even if this worsens the current $L_1^u$ imbalance.
For this, the algorithm first chooses for each vertex $v \in V_i$ the target block $t(v)$ where moving $v$ to $t(v)$ worsens the $L_1^u$ imbalance the least.
Then, we prioritize the vertices based on the following rating (higher is better):
\[
    s(v) \coloneqq \frac{1}{\norm{c(v)}_1} \cdot \frac{c(v)_\ell}{\sum_{\dimindex \ne \ell} c(v)_\dimindex} \cdot \left( 1 + \frac{\gainfn{L_1}(v, t(v))}{\norm{c(v)}_1} \right),
\]
where $\ell = \argmax_{\dimindex \in [d]} c(V_i)_\dimindex$ is the dimension where the current block has maximum weight.
The first term penalizes heavy vertices, as it is usually sufficient to move a relatively small amount of weight.
The second term selects vertices with a high weight in the overloaded dimension, so that the move reduces the internal imbalance of $V_i$ substantially.
Finally, the third term attempts to minimize the increase in $L_1^u$ imbalance
(note that $\gainfn{L_1}(v, t(v))$ is negative and the weight sum $\norm{c(v)}_1$ is an upper bound for its absolute value).
We also consider alternative ratings in Appendix~\ref{sec:app_fallback_variants}, concluding that the exact choice of the rating function has only small influence on the overall performance of the fallback.

The vertex with the best rating is iteratively added to $S_i$ until the accumulated weight is large enough that removing $S_i$ from $V_i$ makes the block balanced.
Then, all selected vertices are moved to their target blocks
%In practice, usually only a small constant number of vertices is selected from each block.
and we apply $L_1^u$ rebalancing again.
If the partition is still imbalanced, the rebalancing stops, as it seems unlikely that balance can be restored.
Note, it is possible that this results in a state with worse imbalance than before the fallback.
We use the same rollback technique as discussed before to restore the state with lowest imbalance.

\section{Multilevel Multi-Constraint Partitioning}\label{sec:implementation}

In the following, we summarize how we integrate multi-constraint support into Mt-KaHyPar.
To keep it short,
we do not discuss algorithmic details of Mt-KaHyPar but refer the interested reader to previous publications~\cite{mt-kahypar-journal, mt-kahypar-d, mt-kahypar-ufm} for additional context.
From a high-level view, we make only minimal semantic changes to the algorithmic components other than the rebalancing;
instead we change the data structures and low-level operations that handle vertex weights.
While our description assumes normalized weights,
the actual implementation normalizes the weights on the fly whenever comparisons between dimensions are necessary.

\subparagraph*{Coarsening and Initial Partitioning.}
Mt-KaHyPar uses size-constrained label propagation clustering, where vertices are visited in random order and greedily assigned to a neighboring cluster that minimizes the heavy-edge rating function~\cite{mt-kahypar-journal}.
A cluster weight limit of $c_{\max} = \frac{c(V)}{160k}$ is enforced to prevent skewed weight distributions in the coarse hypergraph.
We simply enforce this limit separately per constraint, i.e., if a vertex $u$ wants to join a cluster $C \subseteq V$, we require that $c(C)_i + c(u)_i \le \frac{1}{160}$ for each $i \in [d]$.
%This might cause the coarsening to terminate earlier, as the multi-dimensional weight limit is more restrictive for clusters whose weight differs significantly per dimension.
Similarly, when computing an initial solution via recursive bipartitioning of the coarsest hypergraph,
we use the $d$-dimensional balance criterion but leave the portfolio of bipartitioning algorithms otherwise unchanged.
In the multi-constraint setting it can be difficult (or impossible) to find balanced bipartitions~\cite{MulticonstraintMetis};
in such cases we rely on our rebalancing algorithm to restore balance during uncoarsening.

\subparagraph*{Uncoarsening and Refinement.}
While the engineering of refinement algorithms is a topic with many facets, we want to highlight some aspects that are of particular relevance to the multi-constraint setting.
First, refinement algorithms are typically designed with the assumption that the input partition is already balanced.
Consequently, finding a balanced solution early leads to better results in general, as it provides more refinement opportunities.
We therefore apply rebalancing after initial partitioning and after each uncontraction step, if the partition is not yet balanced.
Note that uncoarsening results in a finer hypergraph with
smaller weights, possibly allowing the rebalancing to succeed even if this was impossible before.
Second, refinement algorithms based on iterative vertex moves have two options for escaping local minima:
(i) allow moves that worsen the connectivity, as is done in FM refinement, and
(ii) allow moves that cause a temporary violation of the balance constraint, applying rebalancing afterwards.
In the multi-constraint setting, moves cause a balance violation more easily, which makes it particularly useful to include option (ii).

Therefore, we adapt the unconstrained refinement algorithms of \Partitioner{Mt-KaHyPar}~\cite{mt-kahypar-ufm} to the multi-constraint setting: unconstrained label propagation and unconstrained FM refinement.
The former essentially alternates between rounds of label propagation and rebalancing, and is easily adapted by replacing the previous rebalancing with our new multi-constraint rebalancing.
However, unconstrained FM refinement incorporates two techniques that are difficult to generalize.
First, it estimates penalties for balance-violating moves via a lookup in a weight sequence.
Second, it interleaves sequences of balance-violating moves with balance-restoring moves in a reordering step, possibly discovering additional feasible solutions at intermediate points of the sequence.
Both techniques rely on the monotone behavior of a sequence of scalar weights, which has no direct multi-dimensional equivalent.
Consequently, our current implementation of unconstrained FM does not support these techniques.

\subparagraph*{Postprocessing.}
Mt-KaHyPar initially removes degree zero vertices and reinserts them after the partitioning in a postprocessing step.
In the multi-constraint setting, we need to be careful to not violate the partition balance.
We sort all removed vertices in decreasing order of their $L_1$ weight.
We then iterate over the vertices and assign each vertex $v$ to the block $V_i$ that maximizes $c(v)^\top \cdot \left( \mathbb{1}_d + \varepsilon_d - c(V_i) \right)$.
This ensures that the vertex weight is aligned with the free capacity of $V_i$ and performed well in preliminary experiments.
If necessary, we apply rebalancing afterwards.

\subsection{Scalability and Efficient Implementation}

In principle, the discussed adaptations have no impact on the scalability of Mt-KaHyPar (excluding the new rebalancing),
as we mostly replace one-dimensional operations on weights with $d$-dimensional operations.
If $d$ is a small constant, this does not affect the asymptotic behavior of the algorithm.
However, we need an efficient data layout to minimize overhead in practice.
Mt-KaHyPar stores vertex weights in an array of size $n$ as part of the basic (hyper-)graph data structure.
For $d$ constraints, we instead use an array of size $n \cdot d$ with $d$ consecutive entries for each vertex.
Since operations typically need to access all weight dimensions of a vertex simultaneously, this is a cache-efficient layout.

Many procedures need to perform calculations on vertex weights, and sometimes store the results.
A naive approach might create a new allocation for each intermediate value.
As this would cause significant overhead, we instead combine multiple techniques for a more efficient implementation.
First, we can use pointers for any weights that are already stored elsewhere.
To store results, however, creating a copy is necessary,
and storing partial results can prevent redundant recomputations.
In these cases, we use a pre-allocated copy buffer.
Thereby, only a single initial allocation is required (usually one allocation per thread), which is then reused throughout the procedure.
Furthermore, calculations often involve multiple steps with intermediate results.
For example, to calculate whether moving a vertex would overload the target block, we need the sum of the current block weight and the vertex weight.
In such cases, we do not store intermediate results at all but instead implement the calculation in a single loop, which both saves copies and improves data locality.
This optimization is called \emph{loop fusion} and it is a well-known technique to improve the efficiency of multiple consecutive matrix or vector operations~\cite{fusionViaExpressionTemplates, blisLAFusion}, making it a natural optimization for our use case.\footnote{
    While commonly used compilers such as \texttt{gcc} do not apply loop fusion automatically, we still can avoid doing it manually.
    We implement basic operations only once with a generic programming technique known as \emph{expression templates} in the C++ community~\cite{vandevoorde2002templates, veldhuizen1995expression}, which allows to express calculations naturally while emitting efficient assembly instructions as a single fused loop.
}

\subparagraph*{Atomic Operations.}
Maintaining partition weight limits
in a parallel setting involves atomic operations to synchronize changes from different threads.
Mt-KaHyPar uses \texttt{add-and-fetch} operations, optimistically adding the vertex weight to the block weight while using the return value to cancel the move if the weight limit is violated due to a race condition (restoring the weight via \texttt{sub-and-fetch}).
We use one \texttt{add-and-fetch} per dimension and, in case of constrained refinement, cancel the move if this exceeds the weight limit in the current dimension.
This guarantees that the weight limits are not exceeded, although moves might be canceled unnecessarily due to race conditions.
We do not attempt to prevent this as it should be rare in practice.

\section{Experimental Evaluation}\label{sec:experiments}

We implemented our approach within the Mt-KaHyPar framework (v1.5.3), the source code is available at \url{https://github.com/kahypar/mt-kahypar/tree/esa2026}.
In the following, we evaluate the effectiveness of our rebalancing algorithm and unconstrained refinement (Section~\ref{sec:exp_ablation}) and we compare our approach with the state of the art (Section~\ref{sec:exp_comparison}).
Finally, we show that our algorithm generalizes to larger dimensions (Section~\ref{sec:exp_d6}).

\subsection{Setup}

The code is compiled using gcc 14.2 with flags \texttt{-O3 -march=native}.
We perform all experiments on a machine with an AMD EPYC 9684X processor (one socket with 96 cores), clocked at 2.55-3.7 GHz with 1536 GB
RAM and 1152 MB L3 cache.
We run all parallel algorithms with 96 threads, except where specified otherwise.

\subparagraph*{Benchmark Sets.}
We created our benchmarks~\cite{zenodo-benchmark-data} by adopting sets from the literature and extending the (hyper-)graphs with additional weight dimensions.
Since a common application
is the simultaneous balancing of vertices and edges~\cite{Pulp}, we always use unit weights for the first dimension and the vertex degree as the weight for the second dimension.
Set V consists of 27  hypergraphs derived from VLSI circuit design; namely from the ISPD98 VLSI Circuit Benchmark Suite~\cite{ISPD98} and the DAC 2012 Routability-Driven Placement Contest~\cite{DAC} (the latter instances are significantly larger).
These hypergraphs have natural weights that correspond to the area of electrical components, which we include as a third weight dimension.
Set H comprises 50 hypergraphs from the benchmark set of Heuer and Schlag~\cite{KAHYPAR-CA} (originally 488 instances).
We selected instances with a high variance in vertex degrees, as these result in a more challenging weight distribution.
Since most competitors do not support hypergraph inputs, we also include Set I and Set R from Maas \etal{}~\cite{mt-kahypar-ufm, zenodo-ufm-benchmark}, consisting of large graph instances.
Set I contains graphs with high degree variance, while Set R contains graphs with more regular degrees.
%Note that we use only two weight dimensions for Set H, Set I and Set R (unit weights and vertex degrees).
We exclude the four largest instances from Set I (\texttt{twitter2010}, \texttt{friendster}, \texttt{sk2005}, \texttt{it2004}) since they cause overflows when using 32 bit integers to represent weights.
The benchmark sets and detailed statistics are published at \url{https://zenodo.org/records/19630135}.
See Table~\ref{tab:instances} for a summary of key metrics.

\begin{table}
    \caption{
        Overview of benchmark sets.
        We show pin counts for the smallest and largest hypergraph in the set (graph edges count as two pins),
        the number of instances excluded due to heavy vertices (determined separately for each value of $k$),
        and the instances where Corollary~\ref{corollary:2d} guarantees balance---which is the case if both $d = 2$ and $2 \delta \le \varepsilon$ hold.
    }
    \label{tab:instances}

    \begin{tabular}{llrrrrrr}
        Set & Type & \# & $d$ & Min \# Pins & Max \# Pins & Excluded & Corollary~\ref{corollary:2d} applies\\
        \midrule
        V & hypergraph & 27 & 3 & 50.6\,k & 4,9\,M & 81 (43\,\%) & - \\
        H & hypergraph & 50 & 2 & 3.5\,k & 55.8\,M & 40 (11\,\%) & 101 (28.9\,\%) \\
        I & graph & 34 & 2 & 10.7\,M & 1.7\,B & 6 (3\,\%) & 147 (61.8\,\%) \\
        R & graph & 33 & 2 & 25.4\,M & 1.1\,B & - & 231 (100\,\%) \\
    \end{tabular}
\end{table}

\begin{figure}
    \begin{minipage}{0.49\textwidth}
        \hspace{-8pt}
    \ifpdfplots
    \includegraphics{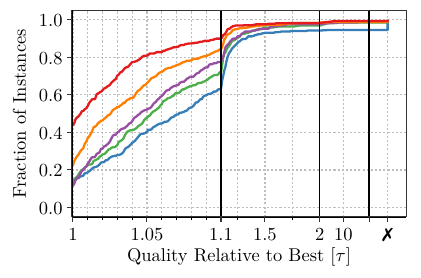}
    \else
    \tikzsetnextfilename{ablation_quality}%
    \input{figures/ablation_quality}%
    \fi

    \end{minipage}
    \hfill
    \begin{minipage}{0.465\textwidth}
        \begin{tabular}{lrrr}
            Algorithm & Conn. & Time & Balanced\\
            \midrule
            Baseline & 1.0 & 1.0 & 94.50\,\% \\
            $L_1^u$ imbalance & 0.996 & 1.118 & 98.56\,\% \\
            +\,rollback & 0.993 & 1.145 & 98.80\,\% \\
            +\,weight red. & 0.964 & 1.099 & 98.80\,\% \\
            +\,fallback & 0.953 & 1.129 & 99.28\,\% \\
        \end{tabular}
    \end{minipage}
    
    \vspace{-5pt}
    \begin{minipage}{\textwidth}
        \hspace{-10pt}
    \ifpdfplots
    \includegraphics{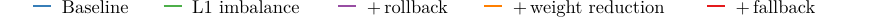}
    \else
    \tikzsetnextfilename{ablation_quality_legend}%
    % Created by tikzDevice version 0.12.6 on 2026-03-24 14:38:33
% !TEX encoding = UTF-8 Unicode
\begin{tikzpicture}[x=1pt,y=1pt]
\definecolor{fillColor}{RGB}{255,255,255}
\path[use as bounding box,fill=fillColor,fill opacity=0.00] (0,0) rectangle (419.17, 10.84);
\begin{scope}
\path[clip] (  0.00,  0.00) rectangle (419.17, 10.84);
\definecolor{drawColor}{RGB}{55,126,184}

\path[draw=drawColor,line width= 1.1pt,line join=round] ( 16.10,  7.92) -- ( 24.78,  7.92);
\end{scope}
\begin{scope}
\path[clip] (  0.00,  0.00) rectangle (419.17, 10.84);
\definecolor{drawColor}{RGB}{77,175,74}

\path[draw=drawColor,line width= 1.1pt,line join=round] ( 78.81,  7.92) -- ( 87.49,  7.92);
\end{scope}
\begin{scope}
\path[clip] (  0.00,  0.00) rectangle (419.17, 10.84);
\definecolor{drawColor}{RGB}{152,78,163}

\path[draw=drawColor,line width= 1.1pt,line join=round] (162.98,  7.92) -- (171.65,  7.92);
\end{scope}
\begin{scope}
\path[clip] (  0.00,  0.00) rectangle (419.17, 10.84);
\definecolor{drawColor}{RGB}{255,127,0}

\path[draw=drawColor,line width= 1.1pt,line join=round] (233.30,  7.92) -- (241.97,  7.92);
\end{scope}
\begin{scope}
\path[clip] (  0.00,  0.00) rectangle (419.17, 10.84);
\definecolor{drawColor}{RGB}{228,26,28}

\path[draw=drawColor,line width= 1.1pt,line join=round] (340.72,  7.92) -- (349.39,  7.92);
\end{scope}
\begin{scope}
\path[clip] (  0.00,  0.00) rectangle (419.17, 10.84);
\definecolor{drawColor}{RGB}{0,0,0}

\node[text=drawColor,anchor=base west,inner sep=0pt, outer sep=0pt, scale=  0.89] at ( 29.86,  4.57) {Baseline};
\end{scope}
\begin{scope}
\path[clip] (  0.00,  0.00) rectangle (419.17, 10.84);
\definecolor{drawColor}{RGB}{0,0,0}

\node[text=drawColor,anchor=base west,inner sep=0pt, outer sep=0pt, scale=  0.89] at ( 92.57,  4.57) {L1 imbalance};
\end{scope}
\begin{scope}
\path[clip] (  0.00,  0.00) rectangle (419.17, 10.84);
\definecolor{drawColor}{RGB}{0,0,0}

\node[text=drawColor,anchor=base west,inner sep=0pt, outer sep=0pt, scale=  0.89] at (176.73,  4.57) {+\,rollback};
\end{scope}
\begin{scope}
\path[clip] (  0.00,  0.00) rectangle (419.17, 10.84);
\definecolor{drawColor}{RGB}{0,0,0}

\node[text=drawColor,anchor=base west,inner sep=0pt, outer sep=0pt, scale=  0.89] at (247.05,  4.57) {+\,weight reduction};
\end{scope}
\begin{scope}
\path[clip] (  0.00,  0.00) rectangle (419.17, 10.84);
\definecolor{drawColor}{RGB}{0,0,0}

\node[text=drawColor,anchor=base west,inner sep=0pt, outer sep=0pt, scale=  0.89] at (354.47,  4.57) {+\,fallback};
\end{scope}
\end{tikzpicture}%
    \fi

    \end{minipage}
    
    \vspace{-8pt}
    \caption{
        Ablation study of the components of our rebalancing algorithm on Set~V and Set~H.
        We compare the resulting connectivity (left) and list geometric means of the connectivity and total running time relative to the baseline, as well as the fraction of balanced results (right).
    }
    \label{fig:ablation}
\end{figure}

\subparagraph*{Methodology.}
In the following evaluation, we allow an imbalance of $\varepsilon = 0.03$;
results for different values are shown in Appendix~\ref{sec:app_other_epsilon}.
We use $k \in \{2, 5, 8, 11, 16, 27, 32\}$ and 5 random seeds for each instance (i.e., combination of hypergraph and $k$),
aggregating connectivity and running time with the arithmetic mean over the seeds.
We consider the result imbalanced only if none of the seeds yielded a balanced solution.
When aggregating over multiple instances, we use the geometric mean.
Reported running times exclude file I/O and parsing.

We use performance profiles~\cite{PERFORMANCE-PROFILES} to compare the solution quality of different algorithms.
A performance profile plots for each algorithm $A$ the fraction of instances where $A$ is within a factor $\tau$ of the best found solution,
i.e., $\frac{1}{|\mathcal{I}|} |\{ q_A(I) \le \tau \cdot \min_{A' \in \mathcal{A}} q_{A'}(I) \mid I \in \mathcal{I} \}|$,
where $\mathcal{I}$ is the set of instances, $\mathcal{A}$ is the set of considered algorithms and $q_A$ the solution quality (connectivity) of algorithm $A$ on a given instance.
We mark infeasible results with \ding{55}.

Depending on the distribution of vertex weights, instances might not admit a balanced solution.
E.g., an instance is trivially infeasible if a singular vertex exceeds the maximum allowed block weight---\nolinebreak
which becomes more likely for larger values of $k$.
To avoid this, we exclude instances with any vertex weight above 70\% of the average block weight.
Note that this is much more relaxed than Corollary~\ref{corollary:2d} (which requires $\delta \le \frac{1}{2}\varepsilon = 1.5\,\%$), leading to the inclusion of many instances where we can not provide a theoretical guarantee (see Table~\ref{tab:instances}).

\subsection{Impact of Rebalancing and Refinement}\label{sec:exp_ablation}

We evaluate the effectiveness of our rebalancing algorithm in an ablation study,
presented in Figure~\ref{fig:ablation}.
We incrementally add the components proposed in Section~\ref{sec:algorithm}, % to the rebalancing,
while all other parts of the algorithm use our final configuration.
The baseline is a multi-constraint version of the standard rebalancing algorithm of \Partitioner{Mt-KaHyPar}, which allows moves only if the target block does not become overloaded.
Using the $L_1^u$ imbalance metric (with $u = 1 + \varepsilon$) reduces the fraction of instances with no balanced solution by a factor of four (1.4\,\% instead of 5.5\,\%),
demonstrating its effectiveness for multi-constraint partitioning.
However, we pay a 12\,\% geometric mean running time overhead since multiple rebalancing rounds can now be necessary.
Adding a rollback mechanism further increases the reliability of finding balanced solutions
and also provides a minor improvement in solution quality.
As suggested by Theorem~\ref{theorem:2d}, the imbalance metric works even better with a slightly reduced target block weight $u$.
Setting a threshold of $t=0.0025$ (see Section~\ref{sec:l1_rebalancing}), this provides a substantial improvement both in solution quality (3\,\%) and running time (4.5\,\%).
%likely because the adjusted metric provides more flexibility to move vertices away from barely overloaded blocks, which leads to faster convergence.
We found in preliminary experiments that the exact choice of $t$ is mostly inconsequential, as long as $t$ is significantly smaller than $\varepsilon$.
Finally, our fallback algorithm further increases the reliability to over 99\,\% and also improves solution quality (1\,\%) at a moderate cost in running time (3\,\%).
%Finally, our fallback algorithm allows to find a balanced solution in even more cases, increasing the reliability to over 99\,\% and further improving solution quality (1\,\%) at a moderate cost in running time (3\,\%).

In Figure~\ref{fig:refinement}, we compare the performance of constrained and unconstrained refinement (see Section~\ref{sec:implementation}), specifically label propagation (LP) and FM as used by default \Partitioner{Mt-KaHyPar}.
We also include hybrid constrained/unconstrained configurations.
The result clearly demonstrates the importance of unconstrained refinement for high-quality partitioning.
All configurations with unconstrained refinement achieve substantially lower connectivity values compared to $-$uLP/$-$uFM (7-10\,\% in the geometric mean), though at the cost of increased running time.
Comparing the unconstrained algorithms, LP performs better than FM, and is already close to the best quality when combined with standard FM---\nolinebreak
unconstrained FM adds only a marginal improvement at a notable running time cost.
This is likely because two techniques are missing % in our FM implementation
compared to the single-constrained case: imbalance penalties and move sequence reordering  (see Section~\ref{sec:implementation}).
For our final configuration, we use the hybrid variant with unconstrained LP refinement.

\begin{figure}
    \begin{minipage}{0.49\textwidth}
        \hspace{-8pt}
    \ifpdfplots
    \includegraphics{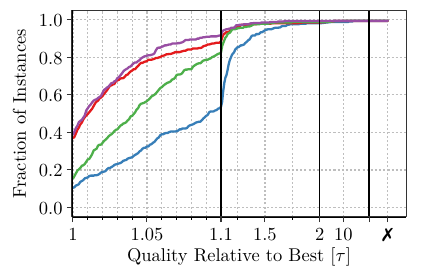}
    \else
    \tikzsetnextfilename{refinement_quality}%
    \input{figures/refinement_quality}%
    \fi

    \end{minipage}
    \hfill
    \begin{minipage}{0.465\textwidth}
        \begin{tabular}{lrrr}
            Algorithm & Conn. & Time & Balanced\\
            \midrule
            $-$uLP/$-$uFM & 1.0 & 1.0 & 99.28\,\% \\
            $+$uLP/$-$uFM & 0.910 & 1.179 & 99.28\,\% \\
            $-$uLP/$+$uFM & 0.933 & 1.115 & 99.52\,\% \\
            $+$uLP/$+$uFM & 0.902 & 1.280 & 99.52\,\% \\
        \end{tabular}
    \end{minipage}
    
    \vspace{-5pt}
    \begin{minipage}{\textwidth}
        \hspace{-10pt}
    \ifpdfplots
    \includegraphics{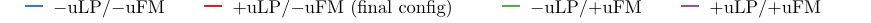}
    \else
    \tikzsetnextfilename{refinement_quality_legend}%
    % Created by tikzDevice version 0.12.6 on 2026-03-24 14:38:46
% !TEX encoding = UTF-8 Unicode
\begin{tikzpicture}[x=1pt,y=1pt]
\definecolor{fillColor}{RGB}{255,255,255}
\path[use as bounding box,fill=fillColor,fill opacity=0.00] (0,0) rectangle (419.17, 10.84);
\begin{scope}
\path[clip] (  0.00,  0.00) rectangle (419.17, 10.84);
\definecolor{drawColor}{RGB}{55,126,184}

\path[draw=drawColor,line width= 1.1pt,line join=round] ( 11.89,  7.92) -- ( 20.57,  7.92);
\end{scope}
\begin{scope}
\path[clip] (  0.00,  0.00) rectangle (419.17, 10.84);
\definecolor{drawColor}{RGB}{228,26,28}

\path[draw=drawColor,line width= 1.1pt,line join=round] ( 98.28,  7.92) -- (106.95,  7.92);
\end{scope}
\begin{scope}
\path[clip] (  0.00,  0.00) rectangle (419.17, 10.84);
\definecolor{drawColor}{RGB}{77,175,74}

\path[draw=drawColor,line width= 1.1pt,line join=round] (241.67,  7.92) -- (250.34,  7.92);
\end{scope}
\begin{scope}
\path[clip] (  0.00,  0.00) rectangle (419.17, 10.84);
\definecolor{drawColor}{RGB}{152,78,163}

\path[draw=drawColor,line width= 1.1pt,line join=round] (328.06,  7.92) -- (336.73,  7.92);
\end{scope}
\begin{scope}
\path[clip] (  0.00,  0.00) rectangle (419.17, 10.84);
\definecolor{drawColor}{RGB}{0,0,0}

\node[text=drawColor,anchor=base west,inner sep=0pt, outer sep=0pt, scale=  0.89] at ( 25.65,  4.57) {$-$uLP/$-$uFM};
\end{scope}
\begin{scope}
\path[clip] (  0.00,  0.00) rectangle (419.17, 10.84);
\definecolor{drawColor}{RGB}{0,0,0}

\node[text=drawColor,anchor=base west,inner sep=0pt, outer sep=0pt, scale=  0.89] at (112.04,  4.57) {$+$uLP/$-$uFM (final config)};
\end{scope}
\begin{scope}
\path[clip] (  0.00,  0.00) rectangle (419.17, 10.84);
\definecolor{drawColor}{RGB}{0,0,0}

\node[text=drawColor,anchor=base west,inner sep=0pt, outer sep=0pt, scale=  0.89] at (255.43,  4.57) {$-$uLP/$+$uFM};
\end{scope}
\begin{scope}
\path[clip] (  0.00,  0.00) rectangle (419.17, 10.84);
\definecolor{drawColor}{RGB}{0,0,0}

\node[text=drawColor,anchor=base west,inner sep=0pt, outer sep=0pt, scale=  0.89] at (341.81,  4.57) {$+$uLP/$+$uFM};
\end{scope}
\end{tikzpicture}%
    \fi

    \end{minipage}
    
    \vspace{-8pt}
    \caption{
        Comparing refinement algorithms on Set~V and Set~H.
        We use either constrained or unconstrained label propagation (uLP) and FM refinement (uFM), resulting in 4 configurations.
        We compare the resulting connectivity (left) and list geometric means of the connectivity and total running time relative to the constrained baseline, as well as the fraction of balanced results (right).
    }
    \label{fig:refinement}
\end{figure}

\begin{figure}
    \begin{minipage}{0.49\textwidth}
        \centering\large
        Set~I
    \end{minipage}
    \hfill
    \begin{minipage}{0.49\textwidth}
        \centering\large
        Set R
    \end{minipage}
    \begin{minipage}{0.49\textwidth}
        \hspace{-8pt}
    \ifpdfplots
    \includegraphics{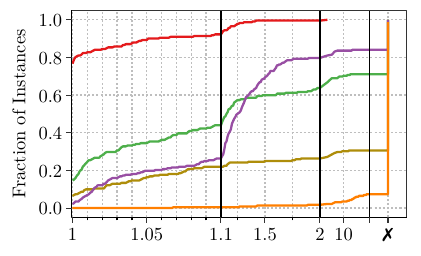}
    \else
    \tikzsetnextfilename{final_quality_ufm_irregular}%
    \input{figures/final_quality_ufm_irregular}%
    \fi

    \end{minipage}
    \hfill
    \begin{minipage}{0.49\textwidth}
        \hspace{-4pt}
    \ifpdfplots
    \includegraphics{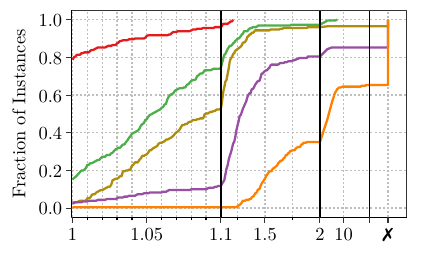}
    \else
    \tikzsetnextfilename{final_quality_ufm_regular}%
    \input{figures/final_quality_ufm_regular}%
    \fi

    \end{minipage}

    \vspace{-12pt}
    \begin{minipage}{0.49\textwidth}
        \centering\large
        Set~H
    \end{minipage}
    \hfill
    \begin{minipage}{0.49\textwidth}
        \centering\large
        Set V
    \end{minipage}
    \begin{minipage}{0.49\textwidth}
        \hspace{-8pt}
    \ifpdfplots
    \includegraphics{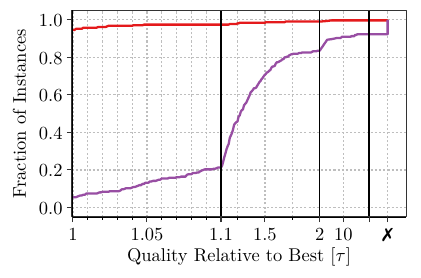}
    \else
    \tikzsetnextfilename{vs_patoh_quality_hg}%
    \input{figures/vs_patoh_quality_hg}%
    \fi

    \end{minipage}
    \hfill
    \begin{minipage}{0.49\textwidth}
        \hspace{-4pt}
    \ifpdfplots
    \includegraphics{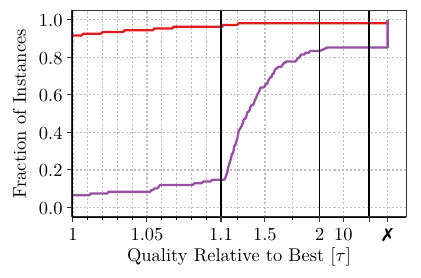}
    \else
    \tikzsetnextfilename{vs_patoh_quality_vlsi}%
    \input{figures/vs_patoh_quality_vlsi}%
    \fi

    \end{minipage}

    \vspace{-10pt}
    \begin{minipage}{\textwidth}
    \ifpdfplots
    \includegraphics{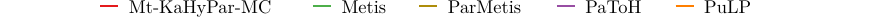}
    \else
    \tikzsetnextfilename{final_quality_legend}%
    % Created by tikzDevice version 0.12.6 on 2026-03-24 14:41:40
% !TEX encoding = UTF-8 Unicode
\begin{tikzpicture}[x=1pt,y=1pt]
\definecolor{fillColor}{RGB}{255,255,255}
\path[use as bounding box,fill=fillColor,fill opacity=0.00] (0,0) rectangle (419.17, 10.84);
\begin{scope}
\path[clip] (  0.00,  0.00) rectangle (419.17, 10.84);
\definecolor{drawColor}{RGB}{228,26,28}

\path[draw=drawColor,line width= 1.1pt,line join=round] ( 48.36,  7.92) -- ( 57.04,  7.92);
\end{scope}
\begin{scope}
\path[clip] (  0.00,  0.00) rectangle (419.17, 10.84);
\definecolor{drawColor}{RGB}{77,175,74}

\path[draw=drawColor,line width= 1.1pt,line join=round] (150.79,  7.92) -- (159.46,  7.92);
\end{scope}
\begin{scope}
\path[clip] (  0.00,  0.00) rectangle (419.17, 10.84);
\definecolor{drawColor}{RGB}{174,141,10}

\path[draw=drawColor,line width= 1.1pt,line join=round] (202.07,  7.92) -- (210.74,  7.92);
\end{scope}
\begin{scope}
\path[clip] (  0.00,  0.00) rectangle (419.17, 10.84);
\definecolor{drawColor}{RGB}{152,78,163}

\path[draw=drawColor,line width= 1.1pt,line join=round] (268.27,  7.92) -- (276.94,  7.92);
\end{scope}
\begin{scope}
\path[clip] (  0.00,  0.00) rectangle (419.17, 10.84);
\definecolor{drawColor}{RGB}{255,127,0}

\path[draw=drawColor,line width= 1.1pt,line join=round] (325.54,  7.92) -- (334.21,  7.92);
\end{scope}
\begin{scope}
\path[clip] (  0.00,  0.00) rectangle (419.17, 10.84);
\definecolor{drawColor}{RGB}{0,0,0}

\node[text=drawColor,anchor=base west,inner sep=0pt, outer sep=0pt, scale=  0.89] at ( 62.12,  4.57) {Mt-KaHyPar-MC};
\end{scope}
\begin{scope}
\path[clip] (  0.00,  0.00) rectangle (419.17, 10.84);
\definecolor{drawColor}{RGB}{0,0,0}

\node[text=drawColor,anchor=base west,inner sep=0pt, outer sep=0pt, scale=  0.89] at (164.54,  4.57) {Metis};
\end{scope}
\begin{scope}
\path[clip] (  0.00,  0.00) rectangle (419.17, 10.84);
\definecolor{drawColor}{RGB}{0,0,0}

\node[text=drawColor,anchor=base west,inner sep=0pt, outer sep=0pt, scale=  0.89] at (215.82,  4.57) {ParMetis};
\end{scope}
\begin{scope}
\path[clip] (  0.00,  0.00) rectangle (419.17, 10.84);
\definecolor{drawColor}{RGB}{0,0,0}

\node[text=drawColor,anchor=base west,inner sep=0pt, outer sep=0pt, scale=  0.89] at (282.02,  4.57) {PaToH};
\end{scope}
\begin{scope}
\path[clip] (  0.00,  0.00) rectangle (419.17, 10.84);
\definecolor{drawColor}{RGB}{0,0,0}

\node[text=drawColor,anchor=base west,inner sep=0pt, outer sep=0pt, scale=  0.89] at (339.30,  4.57) {PuLP};
\end{scope}
\end{tikzpicture}%
    \fi

    \end{minipage}

    \vspace{-15pt}
    \caption{
        Comparing our solution quality to state-of-the-art competitors on graph benchmark sets (top) and hypergraph benchmark sets (bottom).
        Infeasible results or crashes are marked with \ding{55}.
    }
    \label{fig:final}
\end{figure}

\begin{figure}
    \ifpdfplots
    \includegraphics{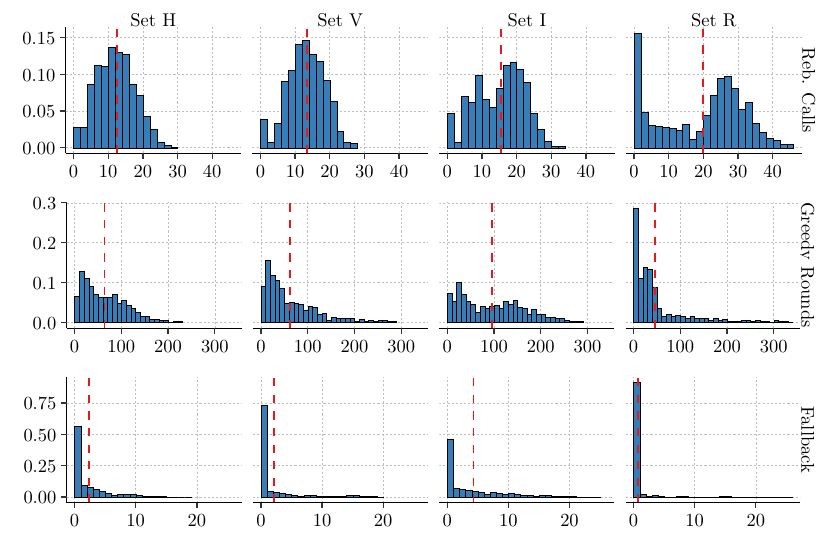}
    \else
    \tikzsetnextfilename{histograms}%
    \input{figures/histograms}%
    \fi

    \vspace{-15pt}
    \caption{
        Histograms of rebalancing calls (top), greedy rebalancing rounds (middle) and fallback invocations during rebalancing (bottom).
        The y-axis depicts the fraction of runs (each seed constitutes a separate data point) and the x-axis the total count per run.
        The mean is marked with a red line.
    }
    \label{fig:histograms}
\end{figure}

\subsection{Comparison to the State of the Art}\label{sec:exp_comparison}

We compare our algorithm (denoted as Mt-KaHyPar-MC), to Metis~\cite{MulticonstraintMetis} (v5.2.1), ParMetis~\cite{MulticonstraintParMetis} (v4.0.3), PaToH~\cite{PATOH} (v3.2) and PuLP~\cite{Pulp} (v0.2), which is a mostly comprehensive list of publicly available partitioners with (partial) multi-constraint support.
For Scotch~\cite{MulticonstraintScotch}, hMetis~\cite{MulticonstraintHMetisTR} and GD~\cite{GradientDescentPartitioning}, we are not aware of a public implementation.
We exclude Zoltan~\cite{ZOLTAN}, as it is limited to geometric partitioning, and TritonPart~\cite{TritonPart}, as the paper provides no data for the claimed multi-constraint support.\footnote{
    The implementation also has non-trivial dependencies and generally a complicated setup process.
}
We evaluated both the direct $k$-way and recursive bisection mode for Metis, and the default and quality configuration for PaToH.
We only show direct $k$-way Metis and default PaToH, as these outperform the alternative configurations (see Appendix~\ref{sec:app_competitors}).
%Since we found that direct $k$-way METIS and default PaToH, respectively, consistently outperform the alternative configurations (see Appendix~\ref{sec:app_competitors}), we only show the former variants.
We use the recommended default settings for all remaining partitioners.%\footnote{
%Note, PaToH supports multi-constraint partitioning only through its C interface, for which we wrote a wrapper that parses the input and calls PaToH. Other partitioners are called via their CLI interface.
%}

\subparagraph{Partition Balance.}
Figure~\ref{fig:final} (top) shows the solution quality and fraction of feasible solutions on Set~I and Set~R.
Mt-KaHyPar-MC always computes a feasible solution, while the competitors have lower reliability, in particular on Set~I (see Table~\ref{tab:time} for exact values).
Specifically, PuLP finds a feasible solution for less than 9\,\% of inputs on Set~I,
likely because it is designed for more relaxed edge balance; the original paper uses 0.5~\cite{Pulp}.
In Appendix~\ref{sec:app_other_epsilon}, we show that the results are similar for both smaller and larger imbalance ($\varepsilon \in \{0.01, 0.1\})$.
Note that the infeasible results include a few crashes in the case of PuLP, PaToH and ParMetis.\footnote{
    PuLP produced a segmentation fault in 30 out of 35 runs on \Instance{rmat-n16m24},
    PaToH had allocation errors on \Instance{uk2005}, \Instance{webbase2001} and \Instance{mawi2015},
    ParMetis crashed on \Instance{uk2005}, \Instance{webbase2001}, \Instance{mawi2015} and \Instance{sinaweibo} and exceeded the time limit on \Instance{rmat-n25m28}.
}

We provide statistics about the behavior of our rebalancing in Figure~\ref{fig:histograms}.
On average, rebalancing the partition once requires between 2.5 (Set~R) and 6 (Set~I) greedy rounds.
On Set~R, rebalancing is used more often (on average 20$\times$ per run), but it tends to succeed with few rounds and without our secondary fallback heuristic.
In general, on more than half of the inputs the fallback is never used, though there is a tail of instances with different behavior.

\subparagraph{Partition Connectivity.}
As shown in Figure~\ref{fig:final} (top),
Mt-KaHyPar-MC finds the best solution of all considered algorithms on 79\,\% of instances for both Set~I and Set~R.
The difference to the next competitor is larger on Set~I (geometric mean: $1.16\times$ versus Metis, $1.27\times$ versus PaToH) than on Set~R ($1.09\times$ versus Metis).
This is consistent with previous results for unconstrained refinement on these instances~\cite{mt-kahypar-ufm}.
Metis tends to find the best results among the competitors, in particular on Set~R.
On Set~I, PaToH finds more balanced solutions than Metis though often with worse connectivity.
We also compare our algorithm to PaToH (the only competitor that supports hypergraph inputs) on Set~H and Set~V in Figure~\ref{fig:final} (bottom), with similar findings:
Mt-KaHyPar-MC achieves significantly better solution quality and more often finds a balanced solution,
even though most instances are outside the guarantees from Corollary~\ref{corollary:2d}.
%Note that differences of few percent are considered significant in the context of (hyper-)graph partitioning, which makes this a substantial improvement.

\begin{table}
    \caption{
        Geometric mean running times on Set~I and Set~R, excluding instances where any algorithm crashed (40 out of 239 on Set~I, none on Set~R), and fraction of balanced results (excluding crashes).
        We also show the running time of our algorithm with 1 and 16 instead of 96 threads.
    }
    \label{tab:time}

    \begin{tabular}{lrrrr}
        Algorithm & Time (Set~I) & Balanced (Set~I) & Time (Set~R) & Balanced (Set~R) \\
        \midrule
        Mt-KaHyPar-MC & 3.50\,s & 100.00\,\% & 1.47\,s & 100.00\,\%\\
        \textit{with 16 threads} & 7.22\,s &  & 3.17\,s & \\
        \textit{with 1 thread} & 61.34\,s & & 27.94\,s & \\
        \midrule
        Metis & 15.60\,s & 71.12\,\% & 7.21\,s & 100.00\,\%\\
        ParMetis & 13.32\,s & 34.98\,\% & 4.40\,s & 96.54\,\%\\
        PaToH & 86.46\,s & 89.86\,\% & 70.38\,s & 85.28\,\%\\
        PuLP & 0.69\,s & 8.77\,\% & 0.50\,s & 65.37\,\%\\
    \end{tabular}
\end{table}

\begin{figure}
    \ifpdfplots
    \includegraphics{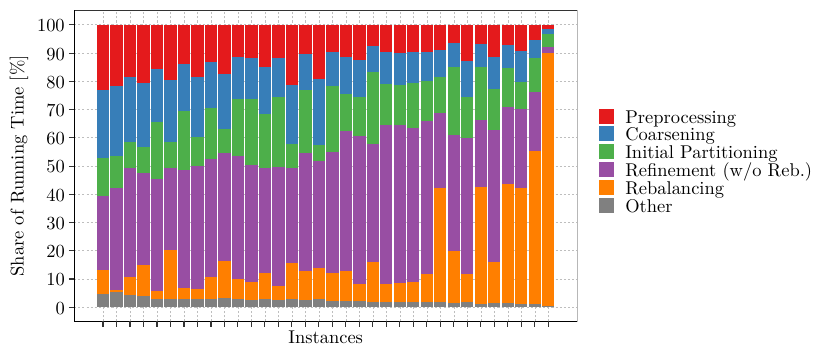}
    \else
    \tikzsetnextfilename{running_time_share_t96}%
    \input{figures/running_time_share_t96}%
    \fi

    \vspace{-10pt}
    \caption{
        Running time of components, as fraction of total time.
        Each bar corresponds to a single instance (hypergraph and $k$) from the larger instances of Set~V, in ascending order of total time.
    }
    \label{fig:time_shares}
\end{figure}

\subparagraph{Running Time.}
While running time is not the primary focus of this work, we aim to preserve the scalability of default Mt-KaHyPar.
The geometric mean running times presented in Table~\ref{tab:time} show that Mt-KaHyPar-MC (with all 96 threads) is the second fastest algorithm after PuLP.
This is unsurprising, as PuLP uses single-level partitioning, trading solution quality for a lower running time.
Compared to all multilevel competitors, i.e., ParMetis with 96 threads, Metis, and PaToH (both single-threaded), our algorithm is still faster when using 16 threads.
However, we also note that using 96 threads only leads to a $18.4\times$ speedup (solution quality is mostly unaffected by the number of threads; see Appendix~\ref{sec:app_threads}),
whereas default Mt-KaHyPar achieves a speedup of $20.5\times$ for 64 threads~\cite{mt-kahypar-journal}.
Figure~\ref{fig:time_shares} provides insight into the running time of different parts of the algorithm (see also Appendix~\ref{sec:app_running_time_shares}).
We can observe a high variance in rebalancing time:
while it is usually below 20\,\% of total time, one instance is dominated by rebalancing.
On this instance, our algorithm could not find a balanced solution,
so the long time is likely due to many failed rebalancing attempts.
This indicates opportunities for future work to reduce the overhead of failed rebalancing attempts and improve the scalability.

\subsection{Results for 6 Dimensions}\label{sec:exp_d6}

\begin{table}
    \caption{Relative increase in geometric mean connectivity and running time on Set~V when adding 3 dimensions with uniformly random weights (Set~V$'$).}
    \label{tab:dimension_6}
    
    \begin{tabular}{l|rr|rr}
        & \multicolumn{2}{c|}{$\varepsilon = 0.03$} & \multicolumn{2}{c}{$\varepsilon = 0.01$}\\
        Algorithm & Connectivity & Running Time & Connectivity & Running Time \\
        \midrule
        Mt-KaHyPar-MC & 1.4\,\% & 4.6\,\% & 10.9\,\% & 9.4\,\% \\
        PaToH & 7.1\,\% & 10.9\,\% & 16.1\,\% & 8.7\,\%
    \end{tabular}
\end{table}

\begin{figure}
    \vspace{5pt}
    \begin{minipage}{0.49\textwidth}
        \centering\large
        $\varepsilon = 0.03$
    \end{minipage}
    \hfill
    \begin{minipage}{0.49\textwidth}
        \centering\large
        $\varepsilon = 0.01$
    \end{minipage}

    \begin{minipage}{0.49\textwidth}
        \hspace{-8pt}
    \ifpdfplots
    \includegraphics{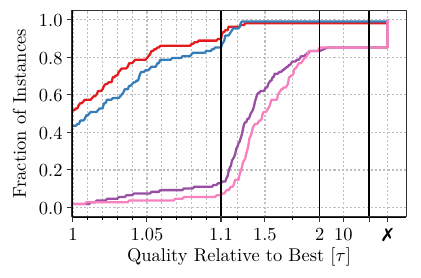}
    \else
    \tikzsetnextfilename{6_dimensions_quality}%
    \input{figures/6_dimensions_quality}%
    \fi

    \end{minipage}
    \hfill
    \begin{minipage}{0.49\textwidth}
    \ifpdfplots
    \includegraphics{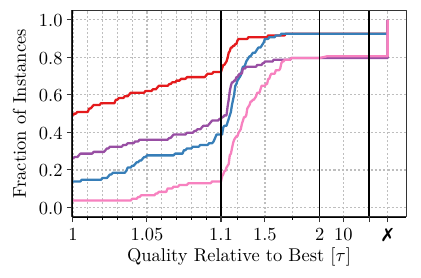}
    \else
    \tikzsetnextfilename{6_dimensions_e1_quality}%
    \input{figures/6_dimensions_e1_quality}%
    \fi

    \end{minipage}
    
    \vspace{-5pt}
    \ifpdfplots
    \includegraphics{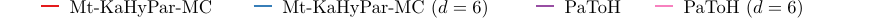}
    \else
    \tikzsetnextfilename{6_dimensions_quality_legend}%
    % Created by tikzDevice version 0.12.6 on 2026-04-02 13:23:04
% !TEX encoding = UTF-8 Unicode
\begin{tikzpicture}[x=1pt,y=1pt]
\definecolor{fillColor}{RGB}{255,255,255}
\path[use as bounding box,fill=fillColor,fill opacity=0.00] (0,0) rectangle (419.17, 10.84);
\begin{scope}
\path[clip] (  0.00,  0.00) rectangle (419.17, 10.84);
\definecolor{drawColor}{RGB}{228,26,28}

\path[draw=drawColor,line width= 1.1pt,line join=round] ( 19.81,  7.92) -- ( 28.48,  7.92);
\end{scope}
\begin{scope}
\path[clip] (  0.00,  0.00) rectangle (419.17, 10.84);
\definecolor{drawColor}{RGB}{55,126,184}

\path[draw=drawColor,line width= 1.1pt,line join=round] (122.23,  7.92) -- (130.90,  7.92);
\end{scope}
\begin{scope}
\path[clip] (  0.00,  0.00) rectangle (419.17, 10.84);
\definecolor{drawColor}{RGB}{152,78,163}

\path[draw=drawColor,line width= 1.1pt,line join=round] (258.32,  7.92) -- (266.99,  7.92);
\end{scope}
\begin{scope}
\path[clip] (  0.00,  0.00) rectangle (419.17, 10.84);
\definecolor{drawColor}{RGB}{247,129,191}

\path[draw=drawColor,line width= 1.1pt,line join=round] (315.59,  7.92) -- (324.26,  7.92);
\end{scope}
\begin{scope}
\path[clip] (  0.00,  0.00) rectangle (419.17, 10.84);
\definecolor{drawColor}{RGB}{0,0,0}

\node[text=drawColor,anchor=base west,inner sep=0pt, outer sep=0pt, scale=  0.89] at ( 33.56,  4.57) {Mt-KaHyPar-MC};
\end{scope}
\begin{scope}
\path[clip] (  0.00,  0.00) rectangle (419.17, 10.84);
\definecolor{drawColor}{RGB}{0,0,0}

\node[text=drawColor,anchor=base west,inner sep=0pt, outer sep=0pt, scale=  0.89] at (135.99,  4.57) {Mt-KaHyPar-MC ($d=6$)};
\end{scope}
\begin{scope}
\path[clip] (  0.00,  0.00) rectangle (419.17, 10.84);
\definecolor{drawColor}{RGB}{0,0,0}

\node[text=drawColor,anchor=base west,inner sep=0pt, outer sep=0pt, scale=  0.89] at (272.08,  4.57) {PaToH};
\end{scope}
\begin{scope}
\path[clip] (  0.00,  0.00) rectangle (419.17, 10.84);
\definecolor{drawColor}{RGB}{0,0,0}

\node[text=drawColor,anchor=base west,inner sep=0pt, outer sep=0pt, scale=  0.89] at (329.35,  4.57) {PaToH ($d=6$)};
\end{scope}
\end{tikzpicture}%
    \fi

    \vspace{-5pt}
    \caption{
        Comparing solution quality on Set~V ($d = 3$) to the solution quality on Set~V$'$ ($d = 6$), using $\varepsilon = 0.03$ (left) and $\varepsilon = 0.01$ (right).
    }
    \label{fig:dimension_6}
\end{figure}

In order to evaluate the reliability of our algorithm for a higher number of dimensions, we additionally created a benchmark set with $d=6$.
Set~V$'$ contains the hypergraphs from Set~V,\nolinebreak
\footnote{We used Set~V since it already has 3 weight dimensions, while our remaining benchmark sets have only 2 weight dimensions---Set~V should therefore result in harder instances.}
however we added 3 new weight dimensions where each dimensions contains weights from a uniform distribution (integers between 1 and 100; note that the range has almost no impact on the difficulty of the instances).
The results are presented in Figure~\ref{fig:dimension_6} and Table~\ref{tab:dimension_6}.
Note that both Mt-KaHyPar-MC and PaToH still find balanced solutions for the same fraction of instances as for $d = 3$, even when using the stricter imbalance value $\varepsilon = 0.01$.
For $\varepsilon = 0.03$, Mt-KaHyPar-MC achieves almost the same solution quality as for Set~V (only 1.4\,\% difference).
We note that uniform random weights result in distributions that are relatively simple to solve;
more adversarial weight distributions could still decrease the reliability of our algorithm.

\section{Conclusion}

We achieve a substantial improvement in quality and reliability over the current state of the art in multi-constraint hypergraph partitioning,
using unconstrained local search algorithms that escape local optima by allowing temporary balance violations.
This is enabled by a new multi-constraint rebalancing algorithm that restores balance with high reliability, providing guarantees for $d=2$ and bounded maximum weight.
Yet, there is room for further improvements.
A multi-constraint variant of the sequence reordering technique for unconstrained FM refinement should be possible, although it might have weaker guarantees.
How to generalize penalties for balance-violating moves is less clear and might require a new kind of approximation.
Furthermore, we plan to investigate the feasibility of strong balance guarantees for $d>2$ and to minimize the running time overhead of rebalancing.

\nikolai{mention that additional changes to coarsening + IP + relaxed refinement constraint are known?}
\nikolai{another possible experiment (future work): single-constraint partitioning with rebalancing as postprocessing}
\review{In some of the applications cited for the multi-constraint use cases, the weights can have different significance and importance. In such a case would it make sense to have different balance parameters for different weights? -> unlikely: this breaks Theorem~\ref{theorem:2d} since we don't find vertices with a fitting weight distribution}
\review{1. Choice of the threshold t: Theorem 2 suggests (...), but practical large weights force a threshold t to effectively reduce $\delta$. How sensitive is the performance to t=0.0025? Would an adaptive scheme that sets t based on the weight distribution be feasible and beneficial?}
\review{3. Connection to vector scheduling: The rebalancing problem without quality consideration is equivalent to vector scheduling. Given that the algorithm provides a guarantee for d=2 (Corollary 3), can this be seen as a practical approximation algorithm for vector scheduling with an additional objective? Clarifying the relationship would be of interest to a broader theory audience.}
\review{The experiments in Appendix F show that the exact choice of rating function has negligible impact. This suggests that the benefit of the fallback comes more from the act of moving a few vertices rather than the specific selection criteria. Could the heuristic be simplified further, e.g., by randomly selecting heavy vertices from the overloaded dimension? Discussing the design implications would be useful.}

%%
%% Bibliography
%%

\newpage
\bibliography{references}

\vfill
%\newpage
\appendix

\section{Counterexample for the Generalization of Theorem~\ref{theorem:2d} to $d > 2$}\label{sec:app_counterexample}

Theorem~\ref{theorem:2d} essentially states that, for $d=2$ and $u \ge 1 + \delta$, it is always possible to move a vertex away from an overloaded block such that the $L_1^u$ imbalance decreases.
In the following, we show that this property of the $L_1^u$ imbalance does not transfer to $d \ge 3$.
For this, we construct an example for $d=3$ where no vertex can be moved in a way that decreases the $L_1^u$ imbalance.
Assume that the goal is to find a $\tau$-balanced partition\footnote{
    Theorem~\ref{theorem:2d} chooses $\tau = u - 1 + \delta$. However, we keep the counterexample more general since the choice might be different for larger $d$.
} using some value of $u$ where $1 \le u \le 1+\tau$.
We use $2n + 1$ blocks that are divided in three sets:
$\{V_1\}, \mathcal{V}_A$ and $\mathcal{V}_B$, where $|\mathcal{V}_A| = |\mathcal{V}_B| = n$.
We choose the following weights for some small $\varepsilon > 0$:
%$c(V_1) = (\tau + \varepsilon, 1, 1)$, $c(V_i) = (u, 2 - u, 2 - u)$ for $i \in \{2, \dots, n + 1\}$ and $c(V_i) = (2 - u + \frac{1}{n}(1 - \tau - \varepsilon), u, u)$ for $i \in \{n + 2, \dots, 2n + 1\}$.
\[
c(V_i) = \begin{cases}
    \left( 1 + \tau + \varepsilon, 1, 1 \right) & \text{if } i=1\\
    \left( 1 + \tau, 1 - \tau, 1 - \tau \right) & \text{if } V_i \in \mathcal{V}_A\\
    \left( 1 - \tau - \frac{1}{n}(\tau + \varepsilon), 1 + \tau, 1 + \tau \right) & \text{if } V_i \in \mathcal{V}_B
\end{cases}
\]

We require $\tau \le 1 - \frac{1}{n + 1} - \varepsilon$.
This ensures that all block weights are non-negative (intermediate step: $(n + 1) \tau + \varepsilon \le n$), making this is a feasible choice of weights.
We can also verify that the summed weight of each dimension is $2n + 1$, i.e., the average weight of the blocks is normalized to 1.
With regards to the individual vertex weights,
we choose them such that $c(v)_1 < c(v)_2 + c(v)_3$ for every $v \in V_1$ (possible due to $c(V_1)_1 = 1 + \tau + \varepsilon < 2 = c(V_1)_2 + c(V_1)_3$).

Regarding the imbalance,
$V_1$ is overloaded by $\varepsilon$ while all other blocks are $\tau$-balanced.
Let us consider whether we can move a vertex $v \in V_1$ while decreasing the $L_1^u$ imbalance.
Moving $v$ to a block in $\mathcal{V}_A$ reduces the $L_1^u$ imbalance of $V_1$ by at most $c(v)_1$, but increases the $L_1^u$ imbalance of the target block by the same amount.
Since the remaining dimensions are below $u$, the overall $L_1^u$ imbalance remains unchanged (or increases).
If we move $v$ to a block in $\mathcal{V}_B$ instead, the change in $L_1^u$ imbalance is at least $c(v)_2 + c(v)_3 - c(v)_1$, which would increase the $L_1^u$ imbalance.

Overall, this excludes analogous statements to Theorem~\ref{theorem:2d}
(decreasing the $L_1^u$  imbalance by moving a single vertex)
for a large range of parameters.
The counterexample works for $\tau$ arbitrarily close to 1, providing a lower bound of twice the average block weight---which constitutes too much imbalance for most practical applications.
Consequently, similar guarantees for $d > 2$ would at least require a different imbalance definition than the $L_1^u$ imbalance.

\newpage
\section{Results for $\varepsilon = 0.01$ and $\varepsilon = 0.1$}\label{sec:app_other_epsilon}

\begin{figure}[h]
    \begin{minipage}{0.49\textwidth}
        \hspace{-8pt}
    \ifpdfplots
    \includegraphics{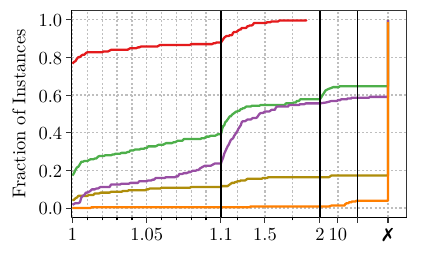}
    \else
    \tikzsetnextfilename{final_e1_quality_ufm_irregular}%
    \input{figures/final_e1_quality_ufm_irregular}%
    \fi

    \end{minipage}
    \hfill
    \begin{minipage}{0.49\textwidth}
    \ifpdfplots
    \includegraphics{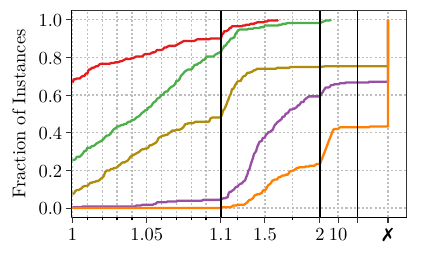}
    \else
    \tikzsetnextfilename{final_e1_quality_ufm_regular}%
    \input{figures/final_e1_quality_ufm_regular}%
    \fi

    \end{minipage}

    \vspace{-8pt}
    \begin{minipage}{0.49\textwidth}
        \hspace{-8pt}
    \ifpdfplots
    \includegraphics{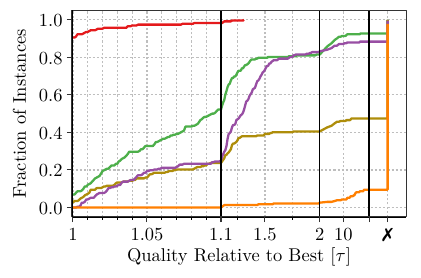}
    \else
    \tikzsetnextfilename{final_e10_quality_ufm_irregular}%
    \input{figures/final_e10_quality_ufm_irregular}%
    \fi

    \end{minipage}
    \hfill
    \begin{minipage}{0.49\textwidth}
    \ifpdfplots
    \includegraphics{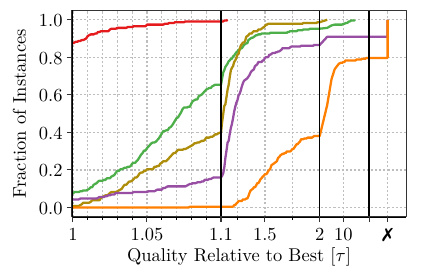}
    \else
    \tikzsetnextfilename{final_e10_quality_ufm_regular}%
    \input{figures/final_e10_quality_ufm_regular}%
    \fi

    \end{minipage}

    \vspace{-5pt}
    \begin{minipage}{\textwidth}
    \ifpdfplots
    \includegraphics{plots/final_quality_legend.pdf}
    \else
    \tikzsetnextfilename{final_quality_legend}%
    % Created by tikzDevice version 0.12.6 on 2026-03-24 14:41:40
% !TEX encoding = UTF-8 Unicode
\begin{tikzpicture}[x=1pt,y=1pt]
\definecolor{fillColor}{RGB}{255,255,255}
\path[use as bounding box,fill=fillColor,fill opacity=0.00] (0,0) rectangle (419.17, 10.84);
\begin{scope}
\path[clip] (  0.00,  0.00) rectangle (419.17, 10.84);
\definecolor{drawColor}{RGB}{228,26,28}

\path[draw=drawColor,line width= 1.1pt,line join=round] ( 48.36,  7.92) -- ( 57.04,  7.92);
\end{scope}
\begin{scope}
\path[clip] (  0.00,  0.00) rectangle (419.17, 10.84);
\definecolor{drawColor}{RGB}{77,175,74}

\path[draw=drawColor,line width= 1.1pt,line join=round] (150.79,  7.92) -- (159.46,  7.92);
\end{scope}
\begin{scope}
\path[clip] (  0.00,  0.00) rectangle (419.17, 10.84);
\definecolor{drawColor}{RGB}{174,141,10}

\path[draw=drawColor,line width= 1.1pt,line join=round] (202.07,  7.92) -- (210.74,  7.92);
\end{scope}
\begin{scope}
\path[clip] (  0.00,  0.00) rectangle (419.17, 10.84);
\definecolor{drawColor}{RGB}{152,78,163}

\path[draw=drawColor,line width= 1.1pt,line join=round] (268.27,  7.92) -- (276.94,  7.92);
\end{scope}
\begin{scope}
\path[clip] (  0.00,  0.00) rectangle (419.17, 10.84);
\definecolor{drawColor}{RGB}{255,127,0}

\path[draw=drawColor,line width= 1.1pt,line join=round] (325.54,  7.92) -- (334.21,  7.92);
\end{scope}
\begin{scope}
\path[clip] (  0.00,  0.00) rectangle (419.17, 10.84);
\definecolor{drawColor}{RGB}{0,0,0}

\node[text=drawColor,anchor=base west,inner sep=0pt, outer sep=0pt, scale=  0.89] at ( 62.12,  4.57) {Mt-KaHyPar-MC};
\end{scope}
\begin{scope}
\path[clip] (  0.00,  0.00) rectangle (419.17, 10.84);
\definecolor{drawColor}{RGB}{0,0,0}

\node[text=drawColor,anchor=base west,inner sep=0pt, outer sep=0pt, scale=  0.89] at (164.54,  4.57) {Metis};
\end{scope}
\begin{scope}
\path[clip] (  0.00,  0.00) rectangle (419.17, 10.84);
\definecolor{drawColor}{RGB}{0,0,0}

\node[text=drawColor,anchor=base west,inner sep=0pt, outer sep=0pt, scale=  0.89] at (215.82,  4.57) {ParMetis};
\end{scope}
\begin{scope}
\path[clip] (  0.00,  0.00) rectangle (419.17, 10.84);
\definecolor{drawColor}{RGB}{0,0,0}

\node[text=drawColor,anchor=base west,inner sep=0pt, outer sep=0pt, scale=  0.89] at (282.02,  4.57) {PaToH};
\end{scope}
\begin{scope}
\path[clip] (  0.00,  0.00) rectangle (419.17, 10.84);
\definecolor{drawColor}{RGB}{0,0,0}

\node[text=drawColor,anchor=base west,inner sep=0pt, outer sep=0pt, scale=  0.89] at (339.30,  4.57) {PuLP};
\end{scope}
\end{tikzpicture}%
    \fi

    \end{minipage}
    
    \vspace{-5pt}
    \caption{
        Comparing our solution quality to state-of-the-art algorithms on Set~I (left) and Set~R (right) for $\varepsilon = 0.01$ (top) and $\varepsilon = 0.1$ (bottom).
        Infeasible solutions or crashes are marked with \ding{55}.
    }
    \label{fig:final_ufm_e1}
\end{figure}

\begin{table}[h!]
    \caption{
        Geometric mean running times on Set~I and Set~R for $\varepsilon = 0.01$ (top) and $\varepsilon = 0.1$ (bottom), excluding instances where any algorithm crashed, as well as fraction of balanced partitions (excluding crashes).
    }
    
    \begin{tabular}{lrrrr}
        Algorithm ($\varepsilon = 0.01$) & Time (Set~I) & Balanced (Set~I) & Time (Set~R) & Balanced (Set~R) \\
        \midrule
        Mt-KaHyPar-MC & 3.41\,s & 100.00\,\% & 1.54\,s & 100.00\,\%\\
        \midrule
        Metis & 15.30\,s & 64.66\,\% & 7.37\,s & 100.00\,\%\\
        ParMetis & 12.14\,s & 20.41\,\% & 4.40\,s & 75.32\,\%\\
        PaToH & 83.40\,s & 63.13\,\% & 69.92\,s & 67.10\,\%\\
        PuLP & 0.64\,s & 5.26\,\% & 0.50\,s & 43.29\,\%\\
    \end{tabular}
    \vspace{\baselineskip}

    \begin{tabular}{lrrrr}
        Algorithm ($\varepsilon = 0.1$) & Time (Set~I) & Balanced (Set~I) & Time (Set~R) & Balanced (Set~R) \\
        \midrule
        Mt-KaHyPar-MC & 3.41\,s & 100.00\,\% & 1.42\,s & 100.00\,\%\\
        \midrule
        Metis & 14.33\,s & 92.67\,\% & 7.07\,s & 100.00\,\%\\
        ParMetis & 13.31\,s & 54.19\,\% & 4.40\,s & 100.00\,\%\\
        PaToH & 82.78\,s & 94.04\,\% & 66.37\,s & 90.91\,\%\\
        PuLP & 0.71\,s & 11.74\,\% & 0.45\,s & 79.65\,\%\\
    \end{tabular}
\end{table}

\newpage

\section{Comparing Metis and PaToH Configurations}\label{sec:app_competitors}

\begin{figure}[h]
    \begin{minipage}{0.49\textwidth}
        \hspace{-8pt}
    \ifpdfplots
    \includegraphics{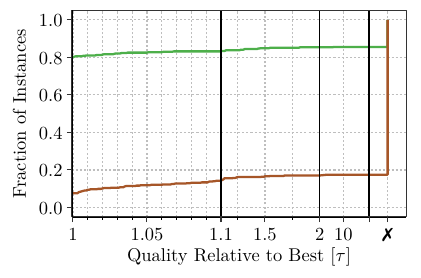}
    \else
    \tikzsetnextfilename{metis_variants_quality}%
    \input{figures/metis_variants_quality}%
    \fi

    \end{minipage}
    \hfill
    \begin{minipage}{0.49\textwidth}
    \ifpdfplots
    \includegraphics{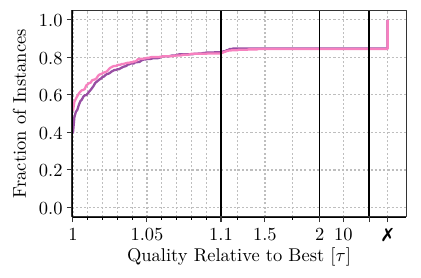}
    \else
    \tikzsetnextfilename{patoh_variants_quality}%
    \input{figures/patoh_variants_quality}%
    \fi

    \end{minipage}
    
    \vspace{-5pt}
    \begin{minipage}{0.49\textwidth}
        \hspace{-8pt}
    \ifpdfplots
    \includegraphics{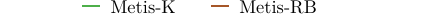}
    \else
    \tikzsetnextfilename{metis_variants_quality_legend}%
    % Created by tikzDevice version 0.12.6 on 2026-03-24 14:33:28
% !TEX encoding = UTF-8 Unicode
\begin{tikzpicture}[x=1pt,y=1pt]
\definecolor{fillColor}{RGB}{255,255,255}
\path[use as bounding box,fill=fillColor,fill opacity=0.00] (0,0) rectangle (202.36, 10.84);
\begin{scope}
\path[clip] (  0.00,  0.00) rectangle (202.36, 10.84);
\definecolor{drawColor}{RGB}{77,175,74}

\path[draw=drawColor,line width= 1.1pt,line join=round] ( 39.51,  7.92) -- ( 48.18,  7.92);
\end{scope}
\begin{scope}
\path[clip] (  0.00,  0.00) rectangle (202.36, 10.84);
\definecolor{drawColor}{RGB}{166,86,40}

\path[draw=drawColor,line width= 1.1pt,line join=round] (101.54,  7.92) -- (110.21,  7.92);
\end{scope}
\begin{scope}
\path[clip] (  0.00,  0.00) rectangle (202.36, 10.84);
\definecolor{drawColor}{RGB}{0,0,0}

\node[text=drawColor,anchor=base west,inner sep=0pt, outer sep=0pt, scale=  0.89] at ( 53.26,  4.57) {Metis-K};
\end{scope}
\begin{scope}
\path[clip] (  0.00,  0.00) rectangle (202.36, 10.84);
\definecolor{drawColor}{RGB}{0,0,0}

\node[text=drawColor,anchor=base west,inner sep=0pt, outer sep=0pt, scale=  0.89] at (115.29,  4.57) {Metis-RB};
\end{scope}
\end{tikzpicture}%
    \fi

    \end{minipage}
    \hfill
    \begin{minipage}{0.49\textwidth}
    \ifpdfplots
    \includegraphics{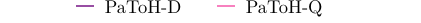}
    \else
    \tikzsetnextfilename{patoh_variants_quality_legend}%
    % Created by tikzDevice version 0.12.6 on 2026-03-24 14:33:24
% !TEX encoding = UTF-8 Unicode
\begin{tikzpicture}[x=1pt,y=1pt]
\definecolor{fillColor}{RGB}{255,255,255}
\path[use as bounding box,fill=fillColor,fill opacity=0.00] (0,0) rectangle (202.36, 10.84);
\begin{scope}
\path[clip] (  0.00,  0.00) rectangle (202.36, 10.84);
\definecolor{drawColor}{RGB}{152,78,163}

\path[draw=drawColor,line width= 1.1pt,line join=round] ( 36.81,  7.92) -- ( 45.48,  7.92);
\end{scope}
\begin{scope}
\path[clip] (  0.00,  0.00) rectangle (202.36, 10.84);
\definecolor{drawColor}{RGB}{247,129,191}

\path[draw=drawColor,line width= 1.1pt,line join=round] (104.69,  7.92) -- (113.37,  7.92);
\end{scope}
\begin{scope}
\path[clip] (  0.00,  0.00) rectangle (202.36, 10.84);
\definecolor{drawColor}{RGB}{0,0,0}

\node[text=drawColor,anchor=base west,inner sep=0pt, outer sep=0pt, scale=  0.89] at ( 50.56,  4.57) {PaToH-D};
\end{scope}
\begin{scope}
\path[clip] (  0.00,  0.00) rectangle (202.36, 10.84);
\definecolor{drawColor}{RGB}{0,0,0}

\node[text=drawColor,anchor=base west,inner sep=0pt, outer sep=0pt, scale=  0.89] at (118.45,  4.57) {PaToH-Q};
\end{scope}
\end{tikzpicture}%
    \fi

    \end{minipage}
    
    \caption{
        Comparing different configurations of Metis (left) and PaToH (right) on the combined instances of Set~I and Set~R for $\varepsilon = 0.03$.
    }
    \label{fig:competitor_variants}
\end{figure}

We compare different configurations of Metis and PaToH in Figure~\ref{fig:competitor_variants}.
Metis has difficulty producing balanced solutions when using recursive bipartitioning, while also being $1.4\times$ slower than the direct $k$-way configuration in the geometric mean.
We note that recursive bipartitioning makes it inherently difficult to find balanced solutions, since each step only considers two blocks at once instead of the global partition balance.
The default and quality configuration of PaToH produce results with nearly identical solution quality.
However, the quality configuration has a $1.19\times$ running time overhead compared to the default configuration.

\section{Quality Impact of Number of Threads}\label{sec:app_threads}

\begin{figure}[h]
    \begin{minipage}{0.49\textwidth}
        \hspace{-8pt}
    \ifpdfplots
    \includegraphics{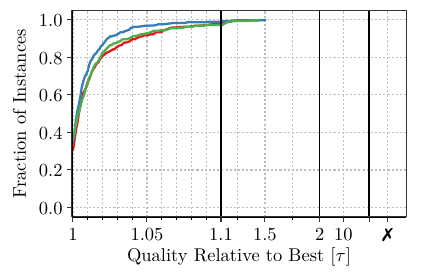}
    \else
    \tikzsetnextfilename{num_threads_quality}%
    \input{figures/num_threads_quality}%
    \fi

    \end{minipage}
    
    \vspace{-5pt}
    \begin{minipage}{0.49\textwidth}
        \hspace{-8pt}
    \ifpdfplots
    \includegraphics{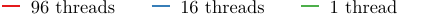}
    \else
    \tikzsetnextfilename{num_threads_quality_legend}%
    % Created by tikzDevice version 0.12.6 on 2026-03-25 14:59:44
% !TEX encoding = UTF-8 Unicode
\begin{tikzpicture}[x=1pt,y=1pt]
\definecolor{fillColor}{RGB}{255,255,255}
\path[use as bounding box,fill=fillColor,fill opacity=0.00] (0,0) rectangle (202.36, 10.84);
\begin{scope}
\path[clip] (  0.00,  0.00) rectangle (202.36, 10.84);
\definecolor{drawColor}{RGB}{228,26,28}

\path[draw=drawColor,line width= 1.1pt,line join=round] (  1.11,  7.92) -- (  9.78,  7.92);
\end{scope}
\begin{scope}
\path[clip] (  0.00,  0.00) rectangle (202.36, 10.84);
\definecolor{drawColor}{RGB}{55,126,184}

\path[draw=drawColor,line width= 1.1pt,line join=round] ( 73.10,  7.92) -- ( 81.77,  7.92);
\end{scope}
\begin{scope}
\path[clip] (  0.00,  0.00) rectangle (202.36, 10.84);
\definecolor{drawColor}{RGB}{77,175,74}

\path[draw=drawColor,line width= 1.1pt,line join=round] (145.08,  7.92) -- (153.76,  7.92);
\end{scope}
\begin{scope}
\path[clip] (  0.00,  0.00) rectangle (202.36, 10.84);
\definecolor{drawColor}{RGB}{0,0,0}

\node[text=drawColor,anchor=base west,inner sep=0pt, outer sep=0pt, scale=  0.89] at ( 14.87,  4.57) {96 threads};
\end{scope}
\begin{scope}
\path[clip] (  0.00,  0.00) rectangle (202.36, 10.84);
\definecolor{drawColor}{RGB}{0,0,0}

\node[text=drawColor,anchor=base west,inner sep=0pt, outer sep=0pt, scale=  0.89] at ( 86.85,  4.57) {16 threads};
\end{scope}
\begin{scope}
\path[clip] (  0.00,  0.00) rectangle (202.36, 10.84);
\definecolor{drawColor}{RGB}{0,0,0}

\node[text=drawColor,anchor=base west,inner sep=0pt, outer sep=0pt, scale=  0.89] at (158.84,  4.57) {1 thread};
\end{scope}
\end{tikzpicture}%
    \fi

    \end{minipage}

    \caption{
        Comparing partition quality for different numbers of threads on the combined instances of Set~I and Set~R for $\varepsilon = 0.03$.
        %Using 16 threads produced slightly better quality than either 96 threads or 1 thread (0.5\,\% in the geometric mean).
    }
\end{figure}

\section{Comparison of Fallback Heuristics}\label{sec:app_fallback_variants}

In Section~\ref{sec:fallback}, we describe a heuristic to resolve cases where the greedy $L_1^u$ rebalancing can not make further progress.
The heuristic selects a small subset of vertices from each overloaded block, which is moved to another block even if this worsens the $L_1^u$ imbalance.
In addition to the rating function $s$ presented in Section~\ref{sec:fallback}, we looked at multiple alternative variants.
%To choose this subset, we evaluated different options for a heuristic rating function, where the best performing variant is presented in Section~\ref{sec:fallback}.
For this, let us consider a generalized rating function that rates a combination of moved vertex $v$ and target block $V_t$.
\[
    s'(v, V_t) \coloneqq \rho(v) \cdot \alpha(v) \cdot \beta(v, V_t)
\]

Here, $\rho$ represent a weight penalty term, $\alpha$ represents a rating based on the current block $V_i$ (i.e., does the internal imbalance decrease)
and $\beta$ represents a rating based on the target block $V_t$ (i.e., how is the internal imbalance and/or $L_1^u$ imbalance of the target block affected).
For these components, we considered the following variants:
\begin{align*}
    \rho_1(v) &= \frac{1}{\norm{c(v)}_1} &
    \rho_2(v) &= 1 &
    & \\
    \alpha_1(v) &= \frac{c(v)_\ell}{\sum_{\dimindex \ne \ell} c(v)_\dimindex} &
    \alpha_2(v) &= c(v)^\top c(V_i) &
    & \\
    \beta_1(v, V_t) &= \frac{\sum_{\dimindex \ne \ell} c(V_t)_\dimindex}{c(V_t)_\ell} &
    \beta_2(v, V_t) &= c(v)^\top (\mathbb{1}_d - c(V_t)) &
    \beta_3(v, V_t) &= 1 + \frac{\gainfn{L_1}(v, V_t)}{\norm{c(v)}_1},
\end{align*}
where $\ell = \argmax_{\dimindex \in [d]} c(V_s)_\dimindex$ is the dimension where $V_i$ has maximum weight.
$\rho_2$ is a variant without penalty for heavy nodes.
$\alpha_1$ and $\beta_1$ prioritize moves where the most overweight dimension of $V_i$ is aligned with the maximum weight dimension of $v$ (see $\alpha_1$) and with the free capacity of $V_t$ (see $\beta_1$).
The idea of $\alpha_2$ and $\beta_2$ is similar, but here we use a dot product to determine whether the weights are aligned.
$\beta_3$ is a bit different, trying to minimize the $L_1^u$ imbalance penalty instead.

The results are presented in Table~\ref{tab:fallback_params}.
Note that the rating function from Section~\ref{sec:fallback} corresponds to $\rho_1$/$\alpha_1$/$\beta_3$.
Overall, the differences are too small to draw any clear conclusion (we observed similar differences for two runs of the same configuration in preliminary experiments).
Therefore, we keep our initial choice of $\rho_1$/$\alpha_1$/$\beta_3$ for the final configuration.
On the other hand, this indicates that the fallback heuristic itself is robust with regards to the exact choice of the rating function.

\begin{table}
    \caption{
        Comparing rating functions for our fallback heuristic on Set~V and Set~H.
        We list geometric means of the connectivity and total running time relative to $\rho_1$/$\alpha_1$/$\beta_1$, as well as the fraction of balanced results.
    }
    \label{tab:fallback_params}

    \begin{minipage}{0.49\textwidth}
        \begin{tabular}{lrrr}
            Algorithm & Conn. & Time & Balanced\\
            \midrule
            $\rho_1$/$\alpha_1$/$\beta_1$ & 1.000 & 1.000 & 99.28\,\% \\
            $\rho_1$/$\alpha_1$/$\beta_2$ & 0.995 & 1.000 & 99.52\,\% \\
            $\rho_1$/$\alpha_1$/$\beta_3$ & 0.999 & 0.990 & 99.28\,\% \\
            $\rho_1$/$\alpha_2$/$\beta_1$ & 1.012 & 0.998 & 99.28\,\% \\
            $\rho_1$/$\alpha_2$/$\beta_2$ & 0.998 & 0.994 & 99.28\,\% \\
            $\rho_1$/$\alpha_2$/$\beta_3$ & 1.001 & 1.000 & 99.52\,\% \\
        \end{tabular}
    \end{minipage}
    \hfill
    \begin{minipage}{0.49\textwidth}
        \begin{tabular}{lrrr}
            Algorithm & Conn. & Time & Balanced\\
            \midrule
            $\rho_2$/$\alpha_1$/$\beta_1$ & 1.013 & 1.001 & 99.52\,\% \\
            $\rho_2$/$\alpha_1$/$\beta_2$ & 1.001 & 1.004 & 99.28\,\% \\
            $\rho_2$/$\alpha_1$/$\beta_3$ & 0.996 & 0.999 & 99.28\,\% \\
            $\rho_2$/$\alpha_2$/$\beta_1$ & 0.998 & 1.002 & 99.28\,\% \\
            $\rho_2$/$\alpha_2$/$\beta_2$ & 0.998 & 0.999 & 99.04\,\% \\
            $\rho_2$/$\alpha_2$/$\beta_3$ & 1.001 & 0.996 & 99.28\,\% \\
        \end{tabular}
    \end{minipage}
\end{table}

\section{Running Time of Algorithm Components}\label{sec:app_running_time_shares}

Figure~\ref{fig:time_shares_t1} shows running time shares of the components of our algorithm in the single-threaded case.
Compared to Figure~\ref{fig:time_shares}, outliers where rebalancing dominates running time are less significant.
This indicates that the rebalancing currently has worse scaling behavior than the remaining algorithm.
Figure~\ref{fig:time_shares_i} and Figure~\ref{fig:time_shares_r} show running time shares on Set~I and Set~R.

\begin{figure}[h]
    \ifpdfplots
    \includegraphics{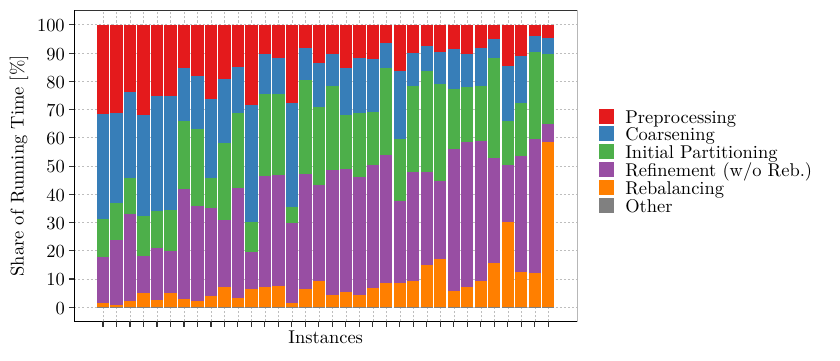}
    \else
    \tikzsetnextfilename{running_time_share_t1}%
    \input{figures/running_time_share_t1}%
    \fi

    \vspace{-10pt}
    \caption{
        Running time of different components of our full algorithm when using a single thread, as fraction of total time.
        Each bar corresponds to a single instance (hypergraph and $k$) from the large instances of Set~V, in ascending order of total running time.
    }
    \label{fig:time_shares_t1}
\end{figure}

\begin{figure}[h]
    \ifpdfplots
    \includegraphics{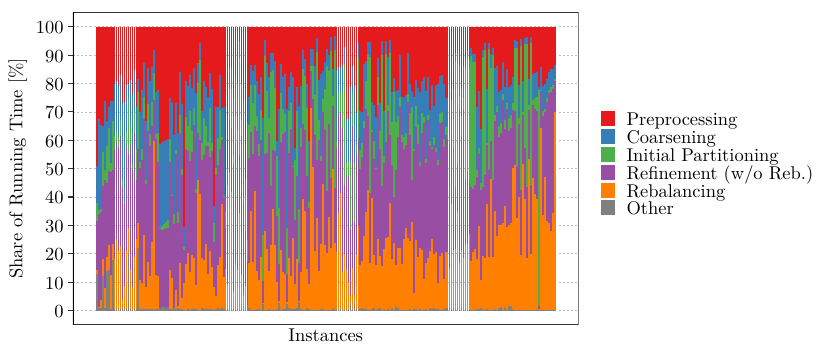}
    \else
    \tikzsetnextfilename{running_time_share_irregular}%
    \input{figures/running_time_share_irregular}%
    \fi

    \vspace{-10pt}
    \caption{
        Running time of different components of our full algorithm on Set~I with 96 threads, as fraction of total time.
        Each bar corresponds to a single instance (graph and $k$), in ascending order of total running time.
    }
    \label{fig:time_shares_i}
\end{figure}

\begin{figure}
    \ifpdfplots
    \includegraphics{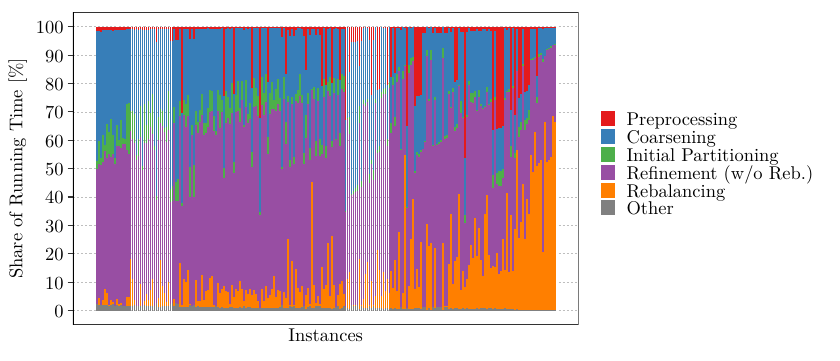}
    \else
    \tikzsetnextfilename{running_time_share_regular}%
    \input{figures/running_time_share_regular}%
    \fi

    \vspace{-10pt}
    \caption{
        Running time of different components of our full algorithm on Set~R with 96 threads, as fraction of total time.
        Each bar corresponds to a single instance (graph and $k$), in ascending order of total running time.
    }
    \label{fig:time_shares_r}
\end{figure}

\end{document}